\documentclass[
times,
aps,
twocolumn,
showpacs,
superscriptaddress,
pre,
floatfix]{revtex4-2}

\usepackage{amsmath,mathtools,amsthm,amssymb,pifont}
\usepackage[utf8]{inputenc}
\usepackage[T1]{fontenc}
\usepackage[american]{babel}
\usepackage{graphicx,bbold,titlesec}
\usepackage{braket}
\usepackage{float}
\usepackage{color}
\usepackage{MnSymbol}
\usepackage{hyperref}
\usepackage[dvipsnames]{xcolor}
\definecolor{BlueRome}{HTML}{4287f5}
\hypersetup{colorlinks=true, linkcolor=BlueRome, citecolor=BlueRome, urlcolor=BlueRome}

\usepackage[normalem]{ulem}
\usepackage{graphicx}
\usepackage{array}
\usepackage{color}
\usepackage{upgreek}
\usepackage{float}
\usepackage{booktabs}
\usepackage{url}
\usepackage{braket}

\usepackage{amsthm}
\usepackage{bm}
\usepackage[T1]{fontenc}
\usepackage{scrextend}
\definecolor{C1}{RGB}{52, 89, 149}
\definecolor{C2}{RGB}{251, 77, 61}
\definecolor{C3}{RGB}{3, 206, 164}
\definecolor{C4}{RGB}{202, 21, 81}

\usepackage{dsfont}
\usepackage{epstopdf}
\usepackage{tikz}
\usepackage{mathrsfs}
\usepackage[export]{adjustbox}
\usepackage{thmtools}
\usepackage{amsmath}
\usepackage{thm-restate}
\usepackage{enumerate}
\usepackage{enumitem}
\usepackage{amssymb}

\usepackage{accents}

\newtheorem{thm}{Theorem}

\newtheorem{lemma}[thm]{Lemma}

\newtheorem{prop}[thm]{Proposition}
\newtheorem{cor}[thm]{Corollary}

\theoremstyle{remark}

\newcommand*{\ot}{\otimes}
\newcommand*{\nn}{\nonumber}

\newcommand*{\id}{\mathds{1}}

\newcommand*{\mc}{\mathcal}
\newcommand*{\dg}{\dagger}

\DeclareMathOperator{\tr}{tr}

\usepackage{mathtools}
\usepackage{thm-restate}

\newenvironment{hproof}{%
  \proof}{\endproof}

\usepackage[export]{adjustbox}
\hypersetup{
pdftitle={Bridging Entanglement and Magic Resources through Operator Space},
pdfsubject={Many-body quantum chaos},
pdfauthor={Neil Dowling},
pdfkeywords={Quantum chaos, Many-body quantum physics, Magic}
}

\begin{document}
\title[]{{Bridging Entanglement and Magic Resources within Operator Space}} 

\author{Neil Dowling}
\email[]{ndowling@uni-koeln.de}
\affiliation{Institut f\"ur Theoretische Physik, Universit\"at zu K\"oln, Z\"ulpicher Strasse 77, 50937 K\"oln, Germany}

\author{Kavan Modi}
\affiliation{School of Physics \& Astronomy, Monash University, Clayton, VIC 3800, Australia}

\author{Gregory A. L. White}
\affiliation{Dahlem Center for Complex Quantum Systems, Freie Universit\"at Berlin, 14195 Berlin, Germany}


\begin{abstract}
    Local-operator entanglement (LOE) dictates the complexity of simulating Heisenberg evolution using tensor network methods, {and bears witness to many-body chaos for local dynamics}. We show that LOE is also sensitive to how non-Clifford a unitary is: its magic resources. In particular, we prove that LOE is always upper-bound by three distinct magic monotones: $T$-count, unitary nullity, and operator stabilizer R\'enyi entropy. Moreover, in the average case for large, random circuits, LOE and magic monotones approximately coincide. 
    Our results imply that an operator evolution that is expensive to simulate using tensor network methods must also be inefficient using both stabilizer and Pauli truncation methods. {In terms of a previous conjecture on the characteristic scaling of LOE, our results also mean that non-integrable spin chains cannot be simulated classically}. Entanglement in operator space therefore measures a unified picture of non-classical resources, in stark contrast to the Schr\"odinger picture. 
\end{abstract}


\maketitle

\textit{Introduction.---} The growth of non-classical resources in quantum dynamics indicates a necessary piece of the puzzle separating classical and quantum simulability. Understanding this separation is essential to a complete characterization of such systems, both from an algorithmic and from a physical perspective.
Arguably the most famous of these resources is entanglement: it both governs the efficiency of tensor network methods~\cite{Verstraete2006,Schuch_2008}, and witnesses quantum critical phase transitions~\cite{Vidal20003,Pasquale_Calabrese_2004,Hastings_2007,Eisert2010area,Silvi2010}. 
Similarly, the concept of non-stabilizerness---or \emph{magic}---has emerged as another key resource of non-classicality. It is well-known that Clifford dynamics can be simulated efficiently on a classical computer by tracking the $N$ stabilizer generators of an $N$-qubit state~\cite{gottesman1998, Aaronson2004}. Including a $T$-gate in the gateset is all that is required to remove these generators and achieve universality---and indeed the best known Clifford+$T$ simulators scale exponentially in the number of $T$-gates~\cite{Pashayan2022,Bravyi2019simulationofquantum}. This departure from stabilizerness can be formalized rigorously in the theory of magic, where $T$-count forms just one example aspect~\cite{PhysRevLett.118.090501,Beverland_2020,Leone2022stab,Liu_2022}. Recently, magic resources have also been found to play a key role in phase transitions~\cite{White2021Conformal, niroula2024phaset, catalano2024magicphase,Fux2024,Bejan2024}, dynamical complexity~\cite{Leone2022stab, garcia_resource_2023, Oliviero2024, gu2024magicinduced, Ahmadi_2024, odavic2024stabiliz}, and (pseudo-)randomness~\cite{Haferkamp_2022, Leone2021quantumchaosis}.

Given the clear role each of these concepts individually play in both quantum information theory and many-body physics,
one would hope to transitively understand the direct connection between the two.
A question of foundational importance is therefore: \textit{How are the two disparate resources of magic and entanglement related to one another?} The solution is not a priori obvious, at least in the usual quantum state setting. For instance, entanglement tends to grow maximally {for} Clifford circuits. On the other hand, product states can have high magic resource. And although details of the entanglement spectrum~\cite{Chamon2014-cb} can serve as witness to magic in a state~\cite{Tirrito2024, Xhek2023flatness}, {this is true only on average and moreover the operational utility of this relation is unclear.}
Even more puzzlingly, there appears no clear connection to the integrability of the underlying dynamics in a many-body setting. Locally-interacting integrable and chaotic models alike produce linearly growing entanglement after a quench~\cite{Calabrese_2005,DeChiara_2006,Prosen2007,Mezei_2017,Nahum2017}, while recent evidence suggests that state magic resources saturate in logarithmic time~\cite{turkeshi2024magicspreadingrandomquantum,tirrito2024anticoncentr,odavi2024stabiliz}. Despite this disparity, there has been recent interest in studying the interplay of, and boundary between, these fundamental resources~\cite{masotllima2024,fux_disentangling_2024,gu2024magicinduced,Frau2024,Bu_2024}.

\begin{figure}[t]
    \centering
    \includegraphics[width=\linewidth]{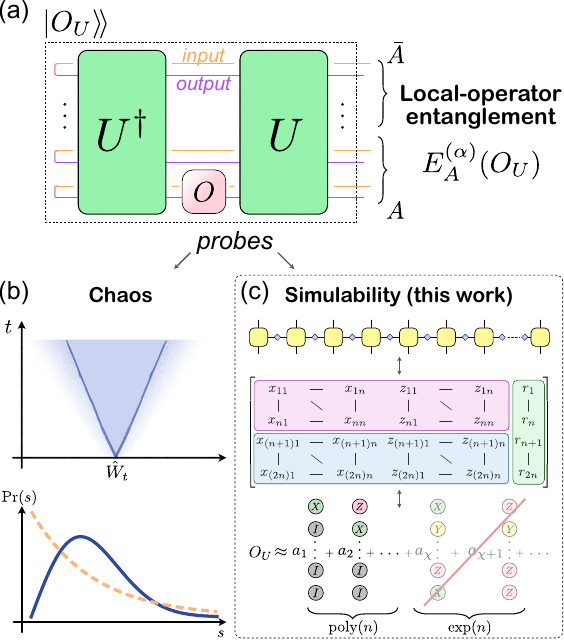}
    \caption{Depiction of key quantities investigated in this work. (a) The local-operator entanglement (LOE) is the entanglement $E_A^{(\alpha)}$ of the Choi state of an operator, $|O_U\rangle\!\rangle$, across some (doubled-space) spatial bipartition $A:\bar{A}$. 
    (b) The extensive scaling of this quantity {has been observed to be related to chaos in many-body systems}, both in the sense of it necessarily indicating information scrambling~\cite{dowling2023scrambling} and {witnessing the non-integrability of spin chains}~\cite{Prosen2007}. (c) In addition to being {sensitive to the entanglement of a unitary (cf. App.~\ref{ap:oe})}, we show that the LOE equally probes magic resources and operator stabilizer entropies, unifying the three corresponding simulation techniques under the one umbrella. 
    }
    \label{fig:main-results}
\end{figure}

{Here, we pose a new understanding of this relationship.} We show that magic is instead conceptually and quantitatively {connected in a direct way to \emph{operator} entanglement. By shifting perspective from states to operators, the key insight is that the free operators in the stabilizer resource theory (Pauli strings) are also free in operator entanglement theory, namely are product operators.} More specifically, consider an operator $O$ in the Heisenberg picture with respect to some evolution $U$, i.e., $O_U= U^\dagger O U $. By the Choi–Jamio\l kowski isomorphism, $O_U$ has a dual pure state, $|O_U\rangle\!\rangle := ( O_U \otimes \id) \ket{\phi^+}$,
where $\ket{\phi^+}$ is the (normalized) maximally entangled state on a doubled Hilbert space. Entanglement and magic resources then have ready generalizations to operators through $|O_U\rangle\!\rangle$. The entanglement $E_A^{(\alpha)}(O_U)$ of this {Choi} state---once appropriately arranged into spatial partitions---is termed the local-operator entanglement (LOE)~\cite{Prosen2007a} {(alternatively `operator-space entanglement entropy')}, and has been studied in a variety of many-body systems~\cite{Prosen2007,Prosen2007a,Prosen2009,Dubail_2017,Jonay2018,Alba2019,Kos2020,Kos2020II,Alba2021}; see 
Fig.~\ref{fig:main-results}.

Our main result can be succinctly summarized as an upper bound on LOE, 
\begin{equation}
    E_A^{(\alpha)}(O_U) \leq m, \label{eq:main1}
\end{equation}
where $m$ refers to any of three different magic monotones: $T$-count, unitary nullity~\cite{Jiang2023}, or operator stabilizer R\'enyi entropy (OSE)~\cite{dowling2024magicheis}. Clearly, entangling dynamics are necessary for LOE growth {(i.e. it is trivially zero for product unitaries{; see App.~\ref{ap:oe}}).} Our result shows intriguingly that it is not sufficient; LOE is limited also by the amount of magic generated by $U$. The quantity therefore represents the interplay between both resources. 
Moreover, we study the typical behavior of LOE over ensembles with restricted magic resources, finding that this upper bound is approximately saturated for such circuits.
The relevant monotones are defined explicitly in Table~\ref{tab:montotones}, while the ensembles where the above bound is saturated are depicted in Fig.~\ref{fig:ensembles}: the well-studied $T$-doped Clifford ensemble~\cite{Haferkamp_2022,Leone2021quantumchaosis,haug2024probingquantumcomplexityuniversal}, and a $\nu$-compressible ensemble which we introduce in this work. 

Eq.~\eqref{eq:main1} is the subject of the remainder of this work, with the full technical version to be found later (Thms.~\ref{thm:main_lower}--\ref{thm:main}). Our result sits in stark contrast to the corresponding state resource theories, where no such relation exists: states with high magic can be unentangled, while zero-magic (stabilizer) states can have maximal entanglement. Beyond the foundational significance of Eq.~\eqref{eq:main1}, there are operational implications. {A range of locally interacting many-body systems} have been observed to have linearly scaling LOE with time, whereas it {tends to} grow at-fastest logarithmically for integrable {spin chains}~\cite{Prosen2007,Prosen2007a,Prosen2009,Dubail_2017,Jonay2018,Alba2019,Kos2020,Kos2020II,Alba2021,ermakov2024polyno}.
{In the context of these past results, Eq.~\eqref{eq:main1} suggests that also an extensive growth of magic resources is a necessary ingredient for {generic} non-integrability.} This means that {`chaotic' (linear)} growth of LOE implies the impossibility of efficient simulation using: tensor network, stabilizer, and Pauli truncation methods (see Cor.~\ref{cor:sim}), serving as a unifying signature of non-classicality. 

\textit{Average Operator Entanglement from Magic.---} Before detailing Thms.~\ref{thm:main_lower}--\ref{thm:main}, we first review the pertinent monotones (summarized in Table~\ref{tab:montotones}) and introduce the unitary ensembles (depicted in Fig.~\ref{fig:ensembles}). Consider an $N$-qubit Hilbert space $\mc{H}$ with total dimension $D=2^N$, and denote by $\mathcal{U}_N$, $\mathcal{C}_N$, and {$\tilde{\mathcal{P}}_N$ the $N$-qubit unitary, Clifford and Pauli groups respectively, with $\mathcal{P}_N = \tilde{\mathcal{P}}_N / \{ \pm i \id\}$ the Pauli strings.} By $U \sim \mathcal{U}_N$, we mean that $U$ is drawn from the unitarily invariant (Haar) measure on $\mathcal{U}_N$, with an equivalent expression for the Cliffords.

Consider a Heisenberg operator $O_U=U^\dagger O U$ for some non-trivial initial Pauli operator $O \in \mc{P}_N \backslash \{\id \}$ and a unitary propagator $U$. The LOE is defined as the bipartite entanglement of the Choi state $|O_U\rangle\!\rangle$ across a chosen bipartition $\mc{H}=\mc{H}_A \otimes \mc{H}_{\bar{A}}$, $
    E^{(\alpha)}_A(O_U) := S^{(\alpha)}( \tr_{\bar{A}}[|O_U\rangle\!\rangle\!\langle\!\langle O_U|] )$, where $S^{(\alpha)}$ represents the (quantum) $\alpha$-R\'enyi entropy, $
    S^{(\alpha)}(\rho) := ({1-\alpha})^{-1}\log(\tr[\rho^\alpha])$.
Through abuse of notation, $\tr_{\bar{A}}$ refers to a partial trace over subsystem ${\bar{A}}$ in the doubled space $\mc{H}_{\bar{A}} \otimes \mc{H}_{\bar{A}}$ with dimension $D_{\bar{A}}^2${, where $\bar{A}$ denotes the complement of $A$}. 
Moreover, the operator purity of $O_U$ is defined in the usual way as $E^{(\text{pur})}_A(O_U) := \exp(-E^{(2)}_A(O_U))$.
To first provide intuition on the interplay between LOE and magic resources, consider Clifford evolution, $C \in \mc{C}_N$. It is immediate to see from the definition of the Clifford group that starting from any initial Pauli operator, $O$ will evolve to some other Pauli operator and so the LOE is preserved to be zero. 
This invites the question of whether a more quantitative relation can be derived: if a unitary has only a few non-Clifford gates, how does the LOE grow? We will address this first by computing the average LOE for ensembles with a tunable magic monotone and then show that there is a precise dependence.

\begin{figure}[t]
    \centering
    \includegraphics[width=\linewidth]{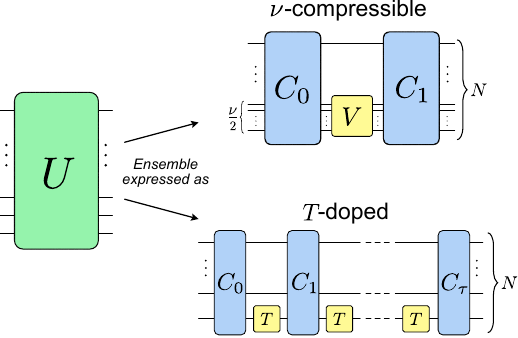}
    \caption{Any unitary can be decomposed as a global Clifford unitaries, $C_0,C_1 \in \mc{C}_N$, sandwiching a non-Clifford unitary on only $\nu(U)/2$ qubits, $V \in \mc{U}_\nu$ (top), where $\nu(U)$ is the unitary nullity~\cite{Jiang2023}. A unitary with $T$-count $\tau(U)$ can be decomposed as layers of $N$-qubit Clifford unitaries $C_i\in\mc{C}_N$ interspersed with $\tau(U)$ single-site $T$-gates (bottom). The $\nu$-compressible ($\mu_\nu$) and $T$-doped ($\mu_\tau$) ensembles are constructed from uniformly sampling these components over the Clifford and unitary groups as appropriate.
    }
    \label{fig:ensembles}
\end{figure}

Consider deep Clifford circuits doped with some non-Clifford gates. That is, random Clifford circuits interspersed by $\tau(U)$ single-site $T$-gates, where $\tau(U)$ is the familiar $T$-count (see Table~\ref{tab:montotones}). If we randomize the Clifford components, such an ensemble of unitaries $\mu_\tau$ is called the $T$-doped Clifford ensemble~\cite{Haferkamp_2022,Leone2021quantumchaosis}, as depicted in Fig.~\ref{fig:ensembles}. Harnessing Weingarten techniques for averaging over the Clifford group~\cite{zhu2016cliffordgroupfailsgracefully,Roth_2018,Leone2021quantumchaosis}, we can readily compute the operator purity for this example (see App.~\ref{ap:tau_lower-bounds}), finding on-average that $\int_{U \sim \mu_\tau} {E^{(\text{pur})}_A(O_U)} = \mc{O}(\exp(-\tau))$. Such a relationship points towards a proportionality between LOE and magic resources. However, {starting from an arbitrary unitary, it is not always possible for it to be compiled with Clifford + $T$ gates exactly. Therefore, }while {conceptually} useful, the $T$-count of a unitary is not {a priori} well-defined. {Moreover, even without this issue, the measure is \textsf{NP}-hard to compute in practice~\cite{vandewetering2024optimisingquantumcircuitsgenerally}}.

A more concrete magic resource stemming from the algebraic structure of the Clifford group is the stabilizer nullity~\cite{Beverland_2020}. For quantum states $\ket{\psi}$, this is defined in terms of the cardinality of its stabilizer group: the number of Pauli operators $P\in\mc{P}_N$ which satisfy $|\langle \psi |P|\psi\rangle | = 1$. The nullity of a unitary operator can be defined analogously in terms of the Choi state $| U \rangle \! \rangle$~\cite{Jiang2023}, $
    \nu(U) := 2N - \log_2 (|\mathrm{Stab}(|{U}\rangle \! \rangle)|)$,
with $\mathrm{Stab}(| U \rangle \! \rangle)=\{ P_1,P_2 \in \mathcal{P}_N: \tr[P_1 U^\dagger P_2 U]/D = \pm1 \}$.
The unitary nullity is a non-negative integer and upper-bounds the state nullity of $U$ acting on any stabilizer state. {Moreover, it gives rise to a powerful representation theorem~\cite{Leone2024learningtdoped,Leone2024learning,gu2024magicinduced}: that any unitary can be decomposed as 
\begin{equation}
    U=C_0 (V_{\ell} \otimes \id_{N-\ell}) C_1, \label{eq:nu-compress-U}
\end{equation}
where $C_0,C_1 \in \mc{C}_N$ may act globally, while $V_{\ell} \in \mc{U}_{\ell }$ acts on $ \nu/2 \leq \ell \leq \nu$ qubits. Here, the lower-bound comes from a counting argument of the unitary stabilizer group generators, while the upper bound is realized explicitly in Ref.~\cite{Leone2024learning}. }

Returning to the question of the LOE's sensitivity to magic resources, we can also compute the average-case operator purity for a unitary with a given unitary nullity $\nu(U)$. We define the $\nu$-compressible ensemble of unitaries $\mu_\nu$ as that generated from independently and uniformly sampling $C_0,C_1 \in \mc{C}_N$ and $V \in \mc{U}_{\ell}$ in Eq.~\eqref{eq:nu-compress-U}, {parametrized by an integer $\ell$. Here, for the Haar random sampling of $V$, the unitary nullity of $U$ will almost surely have nullity $\nu(U)= 2 \ell$. This comes from the invariance of unitary nullity under Clifford composition~\cite{Jiang2023}, and from the fact that Haar random unitaries typically have maximal nullity~\cite{gu2024magicinduced,dowling2024magicheis}: $\nu(U) = \nu(V_{\ell} \otimes \id_{N-\ell}) = \nu(V_{\ell}) \approx 2 \ell$.
In the above, we have also used that nullity is additive under tensor products. Elements of $\mu_\nu$ are therefore typically defined by the decomposition  $U\sim \mu_\nu$: $U=C_0 (V_{\lceil \nu/2 \rceil } \otimes \id_{N-\lceil \nu/2 \rceil }) C_1$, and so we identify $\ell=\nu/2$ in the remainder of this work; see Fig.~\ref{fig:ensembles}.}

Similar to the case of the $T$-doped ensemble, from exact Weingarten techniques we can prove that operator purity on-average scales as $\int_{U \sim \mu_\nu} {E^{(\text{pur})}_A(O_U)} = \mc{O}(\exp(-\nu))$ (see App.~\ref{ap:nu_lowerbound}.

\begin{table}
    \begin{center}
        \begin{tabular}{ c|p{6.2cm} } 
         \textbf{Monotone }& \textbf{Definition} \\ 
         \hline
         LOE & $E^{(\alpha)}_{A}(O_U):= S^{(\alpha)}(\tr_{\bar{A}}[|O_U\rangle\!\rangle\!\langle\!\langle O_U|])$  \\ 
         \hline
         OSE & $M^{(\alpha)}(O_U):= S^{(\alpha)}\!\left(\{\left(\frac{1}{D}\tr[O_U P_i]\right)^{2}\}_{P_i \in \mc{P}_N} \right)$  \\ 
         \hline
         Unitary nullity &{$\nu(U):=2N - \log_2 (|\mathrm{Stab}(|{U}\rangle \! \rangle)|)$}\\ 
         \hline
         $T$-count &  Minimum number $\tau(U)$  
         such that $U=C_0 \prod_{i=1}^{\tau}  (T C_i)$, for $C_i \in \mc{C}_N$.
        \end{tabular}
        \end{center}
        \caption{{Operator entanglement} (first row) and magic (next three rows) monotones that feature in our results. $S^{(\alpha)}$ refers to the $\alpha-$R\'enyi entropy: its quantum and classical version in the first and second row respectively. 
        }
    \label{tab:montotones}
    \end{table}
To clarify the relation of LOE and magic in the two discussed ensembles, we summarize the previous discussion in the following lower-bounds.
\begin{thm} \label{thm:main_lower}
    For R\'enyi indices of $\alpha \leq 2$, and $N_A=N/2$, the average LOE is bounded as
    \begin{equation} \label{eq:lower-bounds}
        \int_{U \sim \mc{E}} E^{(\alpha)}_A(O_U) \geq  \begin{cases}
     \nu -2 - \mc{O}\left( \frac{2^{\nu}}{D} + \frac{1}{4^{\nu}}\right),& \mc{E} = \mu_{\nu},\\
    \log\!\left(\frac{4}{3}\right) \tau - \mc{O}\left( \frac{(4/3)^\tau}{D}\right) ,              & \mc{E} = \mu_{\tau}.
    \end{cases}
    \end{equation}
    Here, the lower-bounds are always positive, and $\mu_\tau$ and $\mu_\nu$ refer to the $T$-doped Clifford and $\nu$-compressible ensembles respectively.
\end{thm}
\begin{hproof}
    We first directly compute the LOE purity using (Clifford/unitary) Weingarten techniques in four-replica space~\cite{collins_integration_2006,zhu2016cliffordgroupfailsgracefully,Roth_2018,Leone2021quantumchaosis}, which led to the previous scalings discussed in the text. Then, after applying Jensen's inequality for the negative logarithm involved in relating purity to $2-$R\'enyi entropy, together with the hierarchy of R\'enyi entropies ($S^{(\alpha)}\leq S^{(\beta)}$ for $\alpha \geq \beta$), we arrive directly at Eq.~\eqref{eq:lower-bounds}. A detailed proof together with the full expressions are given in Apps.~\ref{ap:nu_lowerbound}-\ref{ap:tau_lower-bounds}.
\end{hproof}
{Here, we} see a leading-order linear dependence of LOE with two magic monotones; c.f. Eq.~\eqref{eq:main1}. Interestingly, the dependence on $(3/4)^\tau$ in the $T$-count bound is also observed to leading order for the linear stabilizer R\'enyi entropy of the $T-$doped ensemble~\cite{Leone2022stab,haug2024probingquantumcomplexityuniversal}.

Before moving on to the exact relation between LOE and magic resources, we would like to compare the two lower-bounds of Eq.~\eqref{eq:lower-bounds}. The ensembles $\mu_\tau$ and $\mu_\nu$ coincide only for two cases: $\{\tau=\nu =0\}$, in which case they correspond uniform measure over Clifford group, and $\{ \tau \to \infty, \nu \to 2 N\}$, in which case they correspond to the unitary Haar ensemble. In the former case, as previously discussed, the LOE is trivially equal to zero. In the latter case, we can use our previous results to arrive at the Page-curve for the average LOE purity from Haar random dynamics,
\begin{equation}
    \int_{U \sim \mc{U}_N} E^{(\text{pur})}_{A}(O_U) = \frac{-19 + D + 2 D^2}{(1 + D) (-9 + D^2)}. \label{eq:haar}
\end{equation}
{This result is of independent interest,} improving upon a leading-order expression given in Ref.~\cite{Kudler-Flam2021}.

\textit{Unifying {Operator Entanglement} and Magic.---} 
Having studied the average case, we will now {present exact bounds between LOE and magic} for arbitrary R\'enyi indices and any given evolution, formalizing Eq.~\eqref{eq:main1}. In order to prove our main result, we first review one further magic monotone, operator stabilizer entropy (OSE). The OSE is defined as the entropy of the distribution of squared amplitudes of the Heisenberg operator $O_U$ written in the Pauli basis~\cite{dowling2024magicheis}, $M^{(\alpha)}(O_U) := (1-\alpha)^{-1} \log \sum_{P \in \mathcal{P}} \left(D^{-1}\tr[O_U P]\right)^{2\alpha}$.
This quantity generalizes a popular state measure of magic~\cite{Leone2022stab,Haug_2023,leone2024monotonesmagic}, lower-bounding other magic monotones while satisfying a light-cone bound for local dynamics an dictating the efficiency of Pauli-truncation methods. 
\begin{restatable}{thm}{mainThm} \label{thm:main}
    For any $N$-qubit unitary $U$, any initial operator $O \in \mc{P}_N \backslash \{\id \}$, for any $\alpha \geq 0$, any bipartition $\mc{H}=\mc{H}_A \otimes \mc{H}_{\bar{A}}$, and any unitary propagator $U$, the LOE $\alpha$-R\'enyi entropy satisfies,
    \begin{equation}
         E^{(\alpha)}_A(O_U) \leq {M}^{(\alpha)}(O_U) \leq \nu(U) \leq  \tau(U).\label{eq:main}
    \end{equation}
    Here, $\tau(U)$, $\nu(U)$, and ${M}^{(\alpha)}(O_U)$ are the $T$-count, unitary nullity, and OSE respectively.  
\end{restatable}
\begin{hproof}
    The upper bound is proven through first bounding LOE by OSE, which is apparent from entanglement being the optimum participation entropy over all product bases (with $\mc{P}_N$ being one such basis). For the remaining two upper-bounds, we use a result from Ref.~\cite{Jiang2023} that $\nu(U)\leq \tau(U)$, and make a counting argument to show that ${M}^{(\alpha)}(O_U) \leq \nu(U)$. See App.~\ref{ap:upper_bound} for details.
\end{hproof}
Examining both the upper and lower bounds of LOE in Thms.~\ref{thm:main} and~\ref{thm:main_lower}, we can see that in the average case for large $N \gg \{ \nu,\tau \} \gg 1$, magic resources and {operator} entanglement approximately coincide. It is instructive to compare Eq.~\eqref{eq:main} to the equivalent quantities in the Schr\"odinger picture. Consider the examples of so-called magic states, $\ket{H}^{\otimes N}:=T^{\otimes N}\ket{+ + \dots +}$ and random stabilizer states, $\ket{\psi_{\mathrm{stab}}}:= C\ket{00\dots 0}$ for some deep circuit $C \sim \mathcal{C}_N$. In the first case, the aptly named magic state requires very high magic resources to construct (with a $T$-count and state nullity of $\tau(\ket{T}^{\otimes N})=\nu(\ket{T}^{\otimes N})=N$) but has zero entanglement. On the other hand, a random stabilizer state almost surely has near-maximal entanglement, yet requires no magic resources. The existence of the upper bound Eq.~\eqref{eq:main} in operator space is therefore striking, pointing towards a unified view of non-classical resources in the Heisenberg picture.

{We return to the {relation between LOE and many-body chaos. 
While random matrix theory spectral statistics is one of the most accepted signature of chaos (with some possible caveats, cf. Ref.~\cite{Farshi2023}), dynamical signatures are useful for studying earlier time scales and spatiotemporal correlations in locally interacting many-body systems. One such measure is the scaling in time of LOE: in a range of local many-body systems it has been observed to grow linearly strictly for non-integrable dynamics, while growing at-fastest logarithmically in integrable systems. Beyond numerics~\cite{Prosen2009,Jonay2018,Alba2021}, this conjecture has been confirmed analytically in free~\cite{Prosen2007,Prosen2007a,Dubail_2017}, interacting-integrable~\cite{Alba2019,Kos2020II}, and even chaotic~\cite{Kos2020} spin chains.} 
Moreover, linear LOE growth necessarily implies \emph{scrambling}, as quantified by the decay of out-of-time-ordered correlators (OTOCs)~\cite{dowling2023scrambling}.

We can reinterpret Thm.~\ref{thm:main} through this past literature on the sensitivity of LOE to many-body chaos. Namely, we deduce that the aforementioned non-integrable systems also 
produce at-least $\mc{O}(t)$ non-Clifford resources, while the interacting integrable cases produce at-least $\mc{O}(\log(t))$ magic resources. Note that due to its inherent Lieb-Robinson light-cone, $\mc{O}(t)$ scaling of OSE is maximal for local dynamics~\cite{dowling2024magicheis}. Based on the discussed evidence~\cite{Prosen2007,Prosen2007a,Dubail_2017,Alba2019,Kos2020II,Prosen2009,Jonay2018,Alba2021,dowling2023scrambling}, {generic non-integrable spin-chain dynamics} {produce} at-least linearly growing LOE, OSE, nullity and $T$-count. It is worth comparing this claim to the results of Refs.~\cite{Haferkamp_2022,Leone2021quantumchaosis}. There, it was found that for the $T$-doped Clifford ensemble (cf. Thm.~\ref{thm:main_lower}), $\mc{O}(N)$ $T$-gates are both necessary~\cite{Haferkamp_2022} and sufficient~\cite{Leone2021quantumchaosis}
to reproduce Haar values for higher-order OTOCs. However, these results are valid only on-average over the $T$-doped ensemble, and so leave open the question of what to expect for deterministic many-body dynamics.
Viewing Thm.~\ref{thm:main} in the context of the LOE conjecture complements these results: it assures us that for locally interacting many-body systems, in fact, quantum chaos really cannot be simulated classically.}


To {unpack} the preceding statement, from Thm.~\ref{thm:main} we can deduce a hierarchy of computational complexities using inequivalent methods. We consider three prominent techniques in many-body dynamical simulation: the tensor network method of Heisenberg picture time-evolving block decimation (H-TEBD)~\cite{Hartmann2009} has a computational cost that scales (exponentially) with LOE~\cite{Verstraete2006,Schuch_2008}; stabilizer methods scale (exponentially) with $T$-count according to the Gottesmann-Knill theorem~\cite{gottesman1998,Aaronson2004,Bravyi2019simulationofquantum,Pashayan2022}; and Pauli truncation methods~\cite{Rakovszky2022,lloyd2023ballisticdif,Chan2024,srivatsa2024prob,begusic2024realtime,schuster2024polynomialt} 
have an expense bounded (exponentially) by the OSE~\cite{dowling2024magicheis}.
Note that these resource costs being large do not preclude the efficient simulation of certain features of a system to polynomial precision, such as the anti-concentration of local observables being well-approximated by either: Pauli-truncating an operator to low OSE~\cite{angrisani2024classically}, or analogously using area-law random tensor-networks~\cite{cheng2024pseudoentanglementtensornetworks}. In the following, if the (exponential) resource cost $R$ of a simulation method for some one-parameter unitary $U_t$ scales extensively with time [layers] $t$, $R(U_t) \sim \mc{O}(t)$, we say that the dynamics [circuit] is \emph{non-simulable} with respect to the said method.
\begin{cor} \label{cor:sim}
    If $U_t$ is non-simulable according to H-TEBD, then it is also necessarily non-simulable using Pauli truncation, {which is} further necessarily non-simulable using stabilizer methods.
\end{cor}
{This result means that if} $U_t$ is efficiently simulable according to either stabilizer or Pauli truncation methods, then \emph{it must also be efficient to model using H-TEBD}.

\textit{Conclusion.---}
The operator entanglement of the Heisenberg operator $O_U$ has long been conjectured as a faithful {signature of many-body chaos in locally-interacting systems}, alongside its direct interpretation as tensor network simulability of the dynamics.
Here, we have shown through~Thm.~\ref{thm:main} that this quantity is limited by the minimum of $U$'s entangling capacity and its magic.
The implications here are two-fold: we concretely bridge two seemingly discordant resources of fundamental importance through their {dependence on LOE}. But by extension, our results also integrate a new viewpoint on the non-simulability of complex quantum dynamics. 

We emphasize, however, that this work does not rule out the possibility of classically simulating volume-law LOE through other means. For example, matchgate circuits constitute an alternate class of classically tractable systems. Here, Gaussian observables of free-fermionic systems can be computed efficiently ~\cite{Jozsa2008,Reardon_Smith_2024}. Moreover, the relevant class has been observed to display both high entanglement and magic~\cite{collura2024quantummagicfermionicgaussian}. Nevertheless, we conjecture that a similar relation to the upper bound of Thm~\ref{thm:main} will hold in this setting: that LOE can be bounded by the non-Gaussian resources~\cite{Zhuang_2018,Bu_2024,ermakov2024unifiedf} employed. This intuition is bolstered by the observation that LOE grows at fastest logarithmically with time for free-fermionic Hamiltonian evolution~\cite{Prosen2007,Prosen2007a,Dubail_2017}.

Beyond strict limits, it is important to explore further the behavior of LOE in concrete settings.
One way to understand the average-case saturation in Thm.~\ref{thm:main} is that typical dynamics are maximally entangling. This makes magic the bottleneck, and thus essentially equivalent to the LOE. Less typical scenarios should also be studied, where the relative amounts of each are on equal footing. Appropriate methods on this front include those combining stabilizer and tensor network principles~\cite{masotllima2024,fux_disentangling_2024,qian_augmenting_2024,nakhl2024stabilizer}, suitable for certain systems with an intermediate amount of magic resource. This setting may be well suited to studying LOE in practice, and already interesting regimes have been numerically identified, such as the efficient simulation of $T$-doped Clifford circuits, where $\tau \lesssim N$~\cite{fux_disentangling_2024,nakhl2024stabilizer,huang2024nonstabil}.
    Such an intermediate-amount of magic is also particularly interesting in the context of unitary nullity; cf. Eq.~\eqref{eq:nu-compress-U} and the operational power of $\nu$-compressible states studied in Ref.~\cite{gu2024magicinduced}. Insight on this front may be gleaned from determining classes of unitaries which saturate the respective bounds on compressibility described below Eq.~\eqref{eq:nu-compress-U}. Understanding such restricted complexity dynamical systems will offer further insight into quantum randomness~\cite{Haferkamp_2022,Leone2021quantumchaosis}, simulability~\cite{fux_disentangling_2024,nakhl2024stabilizer}, learnability~\cite{Leone2021choiSRE,Leone2024learning,gu2024magicinduced}, and integrability~\cite{turkeshi2024magicspreadingrandomquantum,tirrito2024anticoncentr,odavi2024stabiliz,dowling2024magicheis}.

Deeply intertwined with the question of dynamical simulability is also that of unitary complexity~\cite{Nielsen2006}.
It is clear that pure-entanglement or pure-magic measures cannot accommodate the linear growth of circuit complexity up to exponential time~\cite{Haferkamp2022}. 
{We lastly remark that our results indicate that LOE may lead to better lower bounds on circuit complexity, given its simultaneous sensitivity to both entanglement and magic.} Exploring generalized extensions of the LOE {to study this observation} is fertile ground for future work.

\twocolumngrid
\begin{acknowledgments}
    The authors thank Pavel Kos, Lorenzo Leone, and Xhek Turkeshi for useful discussions and comments on the manuscript. ND acknowledges funding by the Deutsche Forschungsgemeinschaft (DFG, German Research Foundation) under Germany’s Excellence Strategy - Cluster of Excellence Matter and Light for Quantum Computing (ML4Q) EXC 2004/1 - 390534769. GALW is supported by an Alexander von Humboldt Foundation research fellowship. 
\end{acknowledgments}

%


\newpage

\onecolumngrid

\newpage

\onecolumngrid
\appendix

\section*{Supplemental Material}

\tableofcontents

{
\section{Relation between local-operator entanglement and state entanglement} \label{ap:oe}
{
In this section, we will review the relation between local-operator entanglement (LOE), and operator entanglement of the time evolution operator $U$~\cite{Zanardi2001}. Through this relation, we can then deduce the relationship between LOE and entanglement generation of quantum states. 

Operator entanglement of $U$ is defined in the same way as $O_U$: through the entanglement of the pure Choi state $|U\rangle \! \rangle = U \otimes \id \ket{\phi^+}$; see the discussion above Eq.~\eqref{eq:main1}. Despite a similar formalism in their constructions, LOE and operator entanglement of $U$ behave very differently. For instance, operator entanglement of $U$ typically grows fast for free and non-integrable dynamics alike, unless a system is in a many-body localized phase~\cite{Zhou2017,Dubail_2017}; see the discussion after Thm.~\ref{thm:main} on the behavior of LOE in many-body systems. Moreover, Clifford unitaries may have high operator entanglement, unlike the constant-bounded LOE of Clifford-conjugated initially local operators. We first point out that a product unitary $U$ may generate no LOE. 
\begin{prop}
    Consider a bipartition $\mc{H}=\mc{H}_A \ot \mc{H}_{\bar{A}}$ and some initial operator $O$. Then if for $U_A,\,U_B$ are unitaries acting on $\mc{H}_A$ and $\mc{H}_B$ respectively, 
    \begin{equation}
        E^{(\alpha)}_{A}\left((U_A^\dagger \otimes U_B^\dagger)O(U_A \otimes U_B)\right) = E^{(\alpha)}_{A}(O).
    \end{equation}
\end{prop}
\begin{proof}
    Under the state-operator mapping, for $O_U= (U_A^\dagger \otimes U_B^\dagger)O(U_A \otimes U_B)$ we have that 
    \begin{equation}
        | O_U \rangle \! \rangle = (U_A^* \otimes U_A ) \otimes ((U_B^* \otimes U_B ) ) |O \rangle \! \rangle 
    \end{equation}
    where the tensor product on the right-hand side is between the spatial bipartition of doubled Hilbert spaces labeled as $A$ and $\bar{A}$. As $(U_A^* \otimes U_A ) \otimes ((U_B^* \otimes U_B )$ is a product unitary across this bipartition, the entanglement of $|O \rangle \! \rangle $ cannot change.
\end{proof}
Beyond this, in order for a Heisenberg operator $O_U$ to have some scaling of LOE, it must also have a minimum amount of operator entanglement of $U$. We prove this for $\alpha=0$ R\'enyi entropy (i.e., logarithm of the Schmidt rank), but analogous results also hold for other R\'enyi entropy operator entanglements. 
\begin{prop} \label{prop:5}
    Consider a bipartition $\mc{H}=\mc{H}_A \ot \mc{H}_{\bar{A}}$ and initial operator $O\in \mc{P}_N$. Then if $E^{(0)}_{A}(O_U)=\log(r)$ for $O_U = U^\dagger O U$ and for some constant $r$, then necessarily also,
    \begin{equation}
        E^{(0)}_{A}(U) \geq  \frac{1}{2}\log(r).
    \end{equation}
\end{prop}
\begin{proof}
    Assume that $E^{(0)}_{A}(O_U)=\log(r)$, and assume the converse relation that $ E^{(0)}_{A}(U)=\log(r') <   \frac{1}{2}\log(r)$. We write the matrix-product operator representation of the unitary evolution as $U=\sum_i^r a_i A_i \otimes B_i$. Then from the contraction of $O$ with $U$ and $U^\dagger$, we have  
    \begin{equation}
        O_U = \sum_{i,j}^{r'} a_i^* a_j (A_i \otimes B_i)^{\dagger} (O_A \ot O_B) (A_j \otimes B_j) = \sum_{\ell=1}^{(r')^2} c_\ell \tilde{A}_\ell \otimes \tilde{B}_\ell  \label{eq:mpoProof}
    \end{equation}
    where we recall that $O:= O_A \otimes O_B \in \mc{P}_N$ is product across $A:\bar{A}$, and we have defined a combined index $\ell$. The definitions of other terms in Eq.~\eqref{eq:mpoProof} are clear from context, and we find that the (maximum) bond dimension of $O_U$ is $(r')^2$. This means that $E^{(0)}_{A}(O_U) \leq 2 \log(r') < \log(r)$, which is a contradiction of our initial assumption that $E^{(0)}_{A}(O_U)=\log(r)$. Therefore, if $E^{(0)}_{A}(O_U)=\log(r)$ we know that $E^{(0)}_{A}(U)\geq \log(r)$.
\end{proof}
We are therefore assured that the LOE is sensitive to the operator entanglement of $U$. Finally, we recall the relation between operator entanglement and \textit{entangling power}~\cite{Zanardi2001}, which is the average entanglement generated by a unitary $U$ across random product state inputs~\cite{Zanardi2000},{
\begin{align}
    &e_{P}^{(\alpha)}(U) := \int_{U_A \sim \, \mc{U}_{N_A},U_{\bar{A}} \sim \,  \mc{U}_{N_{\bar{A}}}} \!E^{(\alpha)}_A\left( U(U_A \otimes U_{\bar{A}}) \ket{0}^{\otimes N} \right) \nn\\
    &\,\underset{N_A=N_{\bar{A}}}{=} \frac{D^2}{(D+1)^2} \left(E^{(\alpha)}(U) + E^{(\alpha)}(U \mathbb{S}) - \frac{D^2-1}{D^2} \right),
\end{align}}
where $\mathbb{S}$ is the SWAP unitary between subsystems $\mc{H}_A$ and $\mc{H}_{\bar{A}}$, and the final equality is valid for an equal-sized bipartition. {Through abuse of notation, $E_A^{(\alpha)}(\ket{\psi})$ means the $\alpha$-R\'enyi entanglement entropy of the \textit{state} $\ket{\psi}$ across the bipartition $\mc{H}_A \otimes \mc{H}_{\bar{A}}$}. From this relation together with Prop.~\ref{prop:5}, we can immediately deduce that a unitary which produces a large LOE, also necessarily on-average generates a large amount of entanglement of states. This therefore relates finite LOE to state entanglement, cementing its role as a `bridge' between state entanglement and magic resources. 

}

 \section{Upper-bounding LOE by magic monotones} \label{ap:upper_bound}
 We will here prove the upper bounds of Thm.~\ref{thm:main}. The strategy is to upper-bound the OSE by the unitary nullity, then use that nullity lower-bounds the $T$-count, before proving that OSE further always lower-bounds the LOE of the same R\'enyi index. 
 
We recall the definition of unitary nullity for a given $U$, namely: a unitary $U$ stabilizes $2^{2N-\nu}$ Pauli strings, mapping them to Paulis. This subgroup of Pauli strings is called the unitary stabilizer group of $U$, $\mathrm{Stab}(U)$. Using these definitions, we can immediately relate unitary nullity to the OSE.
\begin{prop} \label{prop:nu-bound}
    For any initial operator $O \in \mathcal{P}_N$ and unitary $U$, \begin{equation}
        {M}^{(\alpha)} (O_U) \leq \nu(U).  \label{eq:nu_bound}
    \end{equation}
\end{prop}
\begin{proof}
    Taking $O \in \mathcal{P}_N$, we consider the conjugate action of $U$
   $O_U = U^\dagger O U$. To derive an upper bound, we consider the `worst case' of $O \in \mc{P}_N\backslash \mathrm{Stab}(U)$. Otherwise, $O_U \in \mc{P}_N$ and OSE is trivially zero. In the non-trivial case, the output $O_U$ is a superposition of at most $2^\nu$ terms (the cardinality of $\mc{P}_N\backslash \mathrm{Stab}(U)$). The highest entropy of the square coefficients of a superposition of $2^\nu$ terms is ${M}^{(\alpha)} (O_U) = \log(2^\nu) = \nu$, corresponding to a uniform superposition. From this optimal case, the bound of Eq.~\eqref{eq:nu_bound} follows directly.
\end{proof}
This result may be of independent interest, as it provides further evidence of the strength of the OSE metric for quantifying the magic of a system. 

Recalling that $\nu(U) \leq \tau(U)$~\cite{Jiang2023}, we have that OSE is always the smallest of the three magic monotones studied in this work (see Table~\ref{tab:montotones}), it remains to prove the relation between OSE and LOE.
\begin{prop}
     \label{prop:osreBound}
    For any R\'enyi index $\alpha$ and across any bipartition $A:\bar{A}$, local-operator entanglement bounds from below the OSE,
    \begin{equation}
        E_A^{(\alpha)}({O_U}) \leq M^{(\alpha)}(O_U).
    \end{equation}
\end{prop}
\begin{proof}
    We first note that the Choi state $|O_U\rangle\!\rangle := (O_U \otimes \id )\ket{\phi^+}$ is a normalized pure state. Then the OSE is just the entropy of the coefficients in the (normalized) computational basis of the Choi state $|O_U\rangle\!\rangle$,
    \begin{align}
        M^{(\alpha)}(O_U) &= \frac{1}{1-\alpha}  \log \sum_{P \in \mathcal{P}} \left(D^{-1}\tr[O_U P]\right)^{2\alpha} \\
        &= \frac{1}{1-\alpha}  \log \sum_{P \in \mathcal{P}} |\langle \! \langle{P|O_U}\rangle\!\rangle|^{2 \alpha}\\
        &= \frac{1}{1-\alpha}  \log \sum_{i=1}^{4^N} |\langle \! \langle{i|O_U}\rangle\!\rangle|^{2 \alpha} \label{eq:OSEoverlap}
    \end{align}
    where the computational basis arises as the Choi states of Pauli operators according to the mapping $P_i \to |{i}\rangle \! \rangle$. {That is, we identify Pauli strings with computational basis states in operator space, and the index $i$ iterates over these basis states of $N$ $4$-dimensional qudits.} Here, we have absorbed the normalization $1/\sqrt{D}$ into the expression of the Choi states $|O_U\rangle\!\rangle$ and $|P\rangle\!\rangle$, such that they are normalized pure states. Consider an arbitrary bipartition $A:B$. Any pure state (including $|O_U\rangle\!\rangle$) can be written as a superposition in terms of some basis $\mathcal{B}$ which is product across this bipartition 
    \begin{equation}
        |O_U\rangle\!\rangle = \sum_{j=1}^{r} \lambda_j |a_j\rangle\!\rangle |b_j\rangle\!\rangle. \label{eq:decomp1}
    \end{equation}
    We call this condition on the basis $\mathcal{B}$ the bipartite-product condition. {Here, the amplitudes $ \lambda_j$ are non-zero for some $1\leq j \leq r$, where $r \leq 4^N$.}
    The {decomposition Eq.~\eqref{eq:decomp1}} is clearly non-unique, and the coefficients $\lambda_j $ together with their cardinality {$r$} are dependent on the choice {of basis} $\mathcal{B}$. If one dephases with respect to this basis, {then} the entropy of the resultant density matrix is 
    \begin{equation}
        S^{(\alpha)}_{\mathcal{B}}(|O_U\rangle\!\rangle) := S^{(\alpha)}(\{ \lambda_j^2 \} ),
    \end{equation}
    where $S^{(\alpha)}$ is the classical $\alpha$-R\'enyi entropy. We call this the bipartite-basis entropy with respect to $\mathcal{B}$. 
    As a basis which is local everywhere (not just in $A:B$), the Pauli basis clearly satisfies the bipartite-product condition, {
    \begin{equation}
        |O_U\rangle\!\rangle = \sum_{i=1}^r \lambda_{i} | i \rangle \! \rangle =\sum_{i_A=1}^{r_A} \sum_{i_B=1}^{r_B} \lambda_{i_A,i_B} |i_A\rangle\!\rangle|i_B\rangle\!\rangle,
    \end{equation}}
    and the corresponding entropy $S^{(\alpha)}_{\mathcal{P}}(|{O_U}\rangle\!\rangle)$ is the OSE [Eq.\eqref{eq:OSEoverlap}]. {Here, the first sum is over $r \leq 4^N$ computational basis elements, and the right hand side is simply writing the computational basis states out in terms of those which act locally on $A$ vs. $B$, where $r=r_A r_B$ and $r_A \leq 4^{N_A}$, $r_B \leq 4^{N_B}$.} The Schmidt basis is defined as the unique basis $\mathcal{B}$ satisfying the bipartite-product condition which minimizes the corresponding bipartite-basis entropy. Then the bipartite-basis entropy of the Schmidt basis is just the entanglement entropy, and so lower bounds all other bipartite-basis entropies, including the OSE. {Note that the above considerations can also be deduced also the definition of the (von Neumann) entropy of the reduced density matrix on $A$: it is the minimum classical Shannon entropy of the square of the diagonal elements of a density matrix, where the minimum is over any basis. The eigenbasis (Schmdit basis) gives the lowest entropy. }
\end{proof}
Combining Props.~\ref{prop:nu-bound} and \ref{prop:osreBound}, together with the fact that $\nu(U) \leq \tau(U)$~\cite{Jiang2023}, we arrive at the hierarchy of upper bounds of Thm.~\ref{thm:main}, 
\begin{equation}
    E_A^{(\alpha)}({O_U}) \leq M^{(\alpha)}(O_U) \leq \nu(U) \leq \tau(U).
\end{equation}

\section{Bounding R\'enyi LOE from average operator purity in replica space} \label{ap:loe_purity}
In this document, we will detail the proofs for the lower bounds of average LOE in terms of magic monotones of nullity and $T$-count, over both the $\nu$-compressible and $T$-doped ensembles. 


The key part of these proofs will be the exact computation of the average operator purity, defined as 
\begin{equation}
    E^{(\text{pur})}(O_U) := \tr[\left(\tr_{\bar{A}}(|O_U\rangle\!\rangle\langle \! \langle {O_U}|)\right)^2], \label{eq:Prepurity}
\end{equation}
i.e. the purity of the Choi stat $|O_U\rangle\!\rangle$. Above, we consider operator entanglement to be across any (spatial) bipartition of $\mc{H}_A \otimes \mc{H}_{\bar{A}}$, of $N_A$ and $N_{\bar{A}}$ qubits respectively. In the case of both unitary ensembles considered here, the particular qubits in $\mc{H}_A$ does not matter, only the size $N_A$. This is because both the Clifford and unitary groups are invariant under left or right multiplication by a permutation matrix (as permutations are contained in the Clifford group, and the Clifford group is contained in the unitary group). We can write the operator purity in replica space as 
\begin{equation}
    E^{(\text{pur})}(O_U) = \frac{1}{D^2}\tr[O_U^{\ot 4} T'], \label{eq:purityLOE}
\end{equation}
where 
\begin{equation}
    T' := T_{(1^A 2^A)(3^A 4^A) (1^{\bar{A}} 4^{\bar{A}})(2^{\bar{A}} 3^{\bar{A}}) } \label{eq:Tprime}
\end{equation}
is a permutation matrix, with $i^A$ denoting the $\mc{H}_A$ subspace of the $i^{th}$ replica space. Eq.~\eqref{eq:purityLOE} is easy to prove graphically, and is part of the standard `replica trick' technique,
\begin{align}
    E^{(\text{pur})}(O_U)=\tr[\left(\tr_{\bar{A}}(|O_U\rangle\!\rangle\!\langle\!\langle O_U|)\right)^2] &= \tr\left[ \frac{1}{D}\tr_{\bar{A}}\left(\includegraphics[scale=1.3, valign=c]{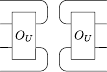}\right) \frac{1}{D}\tr_{\bar{A}}\left(\includegraphics[scale=1.3, valign=c]{parta.pdf}\right)\right] \\
    &=\frac{1}{D^2}\tr\left[ \,\includegraphics[scale=1.3, valign=c]{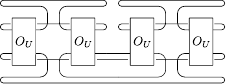}\, \right] \\
    &=\frac{1}{D^2}\tr\left[ \,\includegraphics[scale=1.3, valign=c]{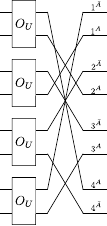}\, \right] = \frac{1}{D^2}\tr[O_U^{\ot 4} T'].
\end{align}
Here, in the first diagram the top two lines represent the (doubled) Hilbert space $\mc{H}_{\bar{A}}$, whereas the bottom two wires represent $\mc{H}_{{A}}$. The normalization of $(1/D^2)$ in Eq.~\eqref{eq:purityLOE} comes from the Choi state normalization in Eq.~\eqref{eq:Prepurity}.

In the replica form of Eq.~\eqref{eq:purityLOE}, the average operator purity takes the form of a $4-$fold average. This can be computed exactly using (generalized) Weingarten calculus. The relevant four-fold unitary (or Clifford) Haar averages contain $(4!)^2=576$ terms (from the cardinality of the permutation group), and so post-simplification we will handle these symbolically using Mathematica. Note that the identity of the initial Pauli $O$ does not matter (local or otherwise), as in the case of both ensembles considered, the first random Clifford applied to it essentially randomizes it. In principle, one could use higher moment formulas to compute higher LOE R\'enyi entropies through the exact expression for the unitary and Clifford commutants~\cite{collins_integration_2006,Gross_2021}. The involved Weingarten calculus becomes much more complex in this case, and we leave it for future work.

From the exact value of the average operator purity (to be computed in the following sections), the proof strategy is to then apply Jensen's inequality to bound the average $2-$R\'enyi entropy, and then to apply the hierarchy of R\'enyi entropy bounds: $S^{(a)} \leq S^{(b)}$ for $a \geq b $, and so a lower bound in terms of $\alpha=2$ is also a lower bound for $\alpha \leq 2$. So in summary 
\begin{equation}
    -\log(\int_U E^{(\text{pur})}(O_U) )\leq \int_U E^{(2)}(O_U) \leq \int_U E^{(1)}(O_U)
\end{equation}
and additionally
\begin{equation}
    -\frac{1}{2}\log(\int_U E^{(\text{pur})}(O_U) )\leq  \int_U E^{(\infty)}(O_U)
\end{equation}
from the bound $E^{(2)} \leq 2 E^{(\infty)}$. This final bound can also be used to provide a lower bound in Thm.~\ref{thm:main_lower} in terms of R\'enyi entropies $\alpha > 2$, but we do not bother as $S^{(\alpha)}$ for $0\leq \alpha \leq 2$ are more informative~\cite{Verstraete2006,Schuch_2008}.

\section{Clifford and Haar averages}
We will use Weingarten techniques to exactly evaluate $4-$fold averages of the operator purity. We have the following expressions for (i) the four-fold Clifford average~\cite{zhu2016cliffordgroupfailsgracefully,Roth_2018,Leone2021quantumchaosis},
\begin{equation}
     \int_{C \in \mc{C}} (C^\dg)^{\ot 4} X C^{\ot 4} = \sum_{\pi, \sigma \in \mathcal{S}_4} \mathrm{Wg}^+_{\pi \sigma} \tr[X \Lambda^+ T_\pi]\Lambda^+ T_\sigma + \mathrm{Wg}^-_{\pi \sigma} \tr[X \Lambda^{- } T_\pi] \Lambda^{- } T_\sigma, \label{eq:cliff_av}
\end{equation}
and (ii) the $4-$fold unitary average~\cite{Mele_2024},
\begin{equation}
    \int_{V \in \mc{U}} (V^\dg)^{\ot 4} X V^{\ot 4} =\sum_{\pi, \sigma \in \mathcal{S}_4} \mathrm{Wg}_{\pi \sigma} \tr[X T_\pi] T_\sigma. \label{eq:haarWein}
\end{equation}
Here, $X \in \mc{H}^{\ot 4}$ lives in $4-$replica space, $\mathcal{S}_4$ is the symmetric group, $T_\pi $ are permutations, $\Lambda^\pm$ are the orthogonal projectors
\begin{equation}
    \Lambda^+ = \frac{1}{D^2} \sum_{P \in \mathcal{P}_N} P^{\otimes 4}; \qquad \Lambda^{-} = \id^{\otimes 4} - \Lambda^+ \label{eq:lambdas}
\end{equation}
which are symmetric under permutations in replica space, and so commute with any $4-$fold Clifford channel. $\mathrm{Wg}^{x}$ are (generalized) Weingarten functions, 
\begin{equation}
    \mathrm{Wg}^{x}_{\pi \sigma} (D) := \sum_{\lambda \vdash 4, D_\lambda^{x} \neq 0} \frac{d_\lambda^2 \chi^{\lambda}(\pi \sigma)}{(4!)^2 D_\lambda^{x}} \label{eq:haar_general}
\end{equation}
where $x \in \{+,-,0 \}$, and 
\begin{itemize}
    \item $\lambda$ labels the irreducible representations of the symmetric group, {with $\lambda \vdash 4$ indicating the integer partitions of $4$,}
    \item $d_\lambda$ is the dimension of the irreducible representation $\lambda$,
    \item $\chi^{\lambda}(\pi \sigma)$ are the characters of the symmetric group, and
    \item $D_\lambda^{x} := \tr[\Lambda^{x} T_\lambda] $, where $T_\lambda$ are projectors onto the irreducible representation $\lambda$.
\end{itemize}
Note that the functions $\mathrm{Wg}^{\pm}$ differ from the usual (Haar) Weingarten functions $\mathrm{Wg}^{0}=:\mathrm{Wg}$ only through the factor $D_\lambda^{\pm}$, so we can use the usual symbolic tools for Haar integrals. Indeed, we will simply use the table from Lemma 1 in Ref.~\cite{zhu2016cliffordgroupfailsgracefully} to determine $D_\lambda^{\pm}$ and $d_\lambda$ in Eq.~\eqref{eq:haar_general}.


As a warm-up before tackling the full computations, we first prove two useful identities.
\begin{lemma}\label{lem:1}
    For $O \in \mc{P}_N$, $T_\pi$ the permutation matrix for $\pi \in \mathcal{S}_4$, and $\Lambda^\pm$ defined as above, 
    \begin{align}
        &\tr[O^{\ot 4} \Lambda^{+ } T_\pi] = \tr[ \Lambda^{+ } T_\pi] =D^{\#(\pi)-2\delta_{\pi o}} , \text{ and} \\
        &\tr[O^{\ot 4} \Lambda^{- } T_\pi] =-\delta_{\pi o}D^{\#(\pi)-2} ,
    \end{align}
    where $\#(\pi)$ is the number of cycles for the permutation $\pi$, and 
    \begin{equation}
    \delta_{\pi o} = \begin{cases}
        0 & \text{for } \pi \in \mc{S}^e \;,\\
        1 & \text{for } \pi \in \mc{S}^o\;.
        \end{cases} \label{eq:deltas}
\end{equation}
$\mc{S}^e \subset \mc{S}_4$ is defined as the `even' permutations, which contain only even cycles, with $\mc{S}^o \subset \mc{S}_4$ being the complement (e.g. $(12)(34)$ is even, but $(12)(3)(4)$ is odd). 
\end{lemma}
\begin{proof}
In $\tr[O^{\ot 4} \Lambda^{\pm } T_\pi] $, both $X=O^{\ot 4}$ and $\Lambda^{\pm}$ are permutation invariant, and so the expression $\tr[X \Lambda^{\pm } T_\pi]$ depends only on the cycle structure of $\pi$. Recalling that $\lambda \vdash 4$ is an integer partition of $4$,
\begin{align}
    \tr[O^{\ot 4} \Lambda^{\pm } T_\pi] &= \pm \frac{1}{D^2} \sum_{P \in \mc{P}} \tr[O^{\ot 4} P^{\ot 4}  T_\pi ] + \delta_{\pm -} \tr[O^{\ot 4} T_\pi ] \nn \\
    &= \pm \frac{1}{D^2} \sum_{P \in \mc{P}} \prod_{c \in \lambda_\pi}\tr[P^c] + \delta_{\pm -} \prod_{c \in \lambda_\pi}\tr[O^c] \label{eq:second}
\end{align}
where we have also used that left multiplication of Pauli strings in a sum over the full group $\mathcal{P}_N$ leaves the sum invariant, as $O\mathcal{P}_N = \mathcal{P}_N$. Note that any phase is canceled out as there are four copies. The trace in the first term of Eq.~\eqref{eq:second} depends only on whether the Pauli is identity or not; where if $P \neq \id$ then this term is non-zero only for even permutations. Moreover, $O$ is a traceless Pauli, so the second term is non-zero only for permutations $\mathcal{S}^e$ solely with cycles of even length. We therefore arrive at 
\begin{align}
    \tr[O^{\ot 4} \Lambda^{\pm } T_\pi] &=\begin{cases}
        \pm ({D^2}/{D^2}) \prod_{c \in \lambda_\pi}\tr[\id] + \delta_{\pm -} \prod_{c \in \lambda_\pi}\tr[\id] & \text{for }\lambda_\pi \ni c =2k\;,\\
        \pm (1/{D^2})\prod_{c \in \lambda_\pi}\tr[\id]  & \text{otherwise}\;,
        \end{cases}\\
        &=\begin{cases}
        \pm D^{\#(\pi)} + \delta_{\pm -} D^{\#(\pi)}  & \text{for }\lambda_\pi \ni c =2k\; (\pi \in \mathcal{S}^e) ,\\
        \pm D^{\#(\pi)-2}  & \text{otherwise } (\pi \in \mathcal{S}^o)\;. 
        \end{cases}\label{eq:CliffWg1} 
\end{align}
Where $\#(\pi)$ is the number of cycles for the permutation $\pi$. Simplifying this further, 
\begin{equation}
    \tr[O^{\ot 4} \Lambda^{+ } T_\pi] =(D^{\#(\pi)-2\delta_{\pi o}} ),
\end{equation}
and 
\begin{equation}
    \tr[O^{\ot 4} \Lambda^{- } T_\pi] =-\delta_{\pi o}(D^{\#(\pi)-2} ),
\end{equation}
where $  \delta_{\pi o} $ is defined in Eq.~\eqref{eq:deltas}.
\end{proof}

\begin{lemma} \label{lem:2}
    For a single qubit $Z-$rotation $T$-gate $P_{\pi/4}=\ket{0}\bra{0} + \exp[i\pi/4 ] \ket{1}\bra{1}$ and $K=P_\theta \otimes \id_{N-1}$, then 
    \begin{equation}
        \mathrm{tr} \left(  K^{\otimes 4} \Lambda^+ K^{\dagger \otimes 4} T_{\sigma } \right)=D^{\#(\sigma ) -2 \delta_{\sigma o}},\label{eq:1st}
    \end{equation}
    and, 
     \begin{align}
        \mathrm{tr} \left(  K^{\otimes 4} \Lambda^+ K^{\dagger \otimes 4} \Lambda^{+} T_{\sigma } \right) &= D^{\#(\sigma ) -2 \delta_{\sigma o}}_{N-1} \frac{1}{2^4}\sum_{P_1,P_2 \in \mc{P}_1}\tr[P_\theta^{\otimes 4} P_1 P_\theta^{\dagger \otimes 4} P_2 T_{\sigma } ]\\
        &=: D^{\#(\sigma ) -2 \delta_{\sigma o}}_{N-1} f_\sigma .\label{eq:2nd}
    \end{align}
    This can be directly generalized to other single-qubit non-Cliffords $P_\theta$. 
\end{lemma}
\begin{proof}
     For the first identity \eqref{eq:1st}, we have 
    \begin{align}
        \mathrm{tr} \left(  K^{\otimes 4} \Lambda^+ K^{\dagger \otimes 4} T_{\sigma } \right)&= \mathrm{tr} \left(  (P_\theta \otimes \id_{N-1})^{\otimes 4} \Lambda^+ (P_\theta \otimes \id_{N-1})^{\dagger \otimes 4} T_{\sigma } \right) \\
        &=\frac{1}{2^2}\sum_{P_1\in \mathcal{P}_1}  \tr[P_\theta^{\ot 4} P_1^{\ot 4} (P_\theta^\dagger)^{\ot 4} T_{\sigma } ] \tr[ \Lambda^+_{N-1}  T_{\sigma} ] \\
        &= D^{\#(\sigma ) -2 \delta_{\sigma o}}_{1}D^{\#(\sigma) -2 \delta_{\sigma o}}_{N-1} = D^{\#(\sigma ) -2 \delta_{\sigma o}}.
    \end{align}
    where we have again used that $\Lambda^+$ decomposes into a product operator in replica space, and that $\tr[ \Lambda^+   T_{\sigma} ] = D^{\#(\sigma) -2 \delta_{\sigma o}}$. Furthermore, for a function $\tr[X^{\ot 4} T_{\sigma}]$, only the trace properties of $X$ matter, and conjugation by the unitary $P_\theta$ preserves the trace (and that $(P_\theta P_1 P_\theta^\dagger)^{2k} = \id$). The proof for the second identity \eqref{eq:2nd} first follows a similar method to above, 
\begin{align}
    \mathrm{tr} \left(  K^{\otimes 4} \Lambda^+ K^{\dagger \otimes 4} \Lambda^{+} T_{\sigma } \right) &=\frac{1}{2^4}\sum_{P_1,P_2\in \mathcal{P}_1}  \tr[P_\theta^{\ot 4} P_1^{\ot 4} (P_\theta^\dagger)^{\ot 4} P_2^{\ot 4} T_{\sigma } ] \tr[ \Lambda^+_{N-1}   \Lambda^+_{N-1} T_{\sigma } ]   \\
    &=\frac{1}{2^4}\sum_{P_1,P_2\in \mathcal{P}_1}  \tr[P_\theta^{\ot 4} P_1^{\ot 4} (P_\theta^\dagger)^{\ot 4} P_2^{\ot 4} T_{\sigma } ]   D^{\#(\sigma ) -2 \delta_{(\sigma) o}}_{N-1},
\end{align}
where we have again used that $\Lambda^+$ decomposes into a product operator, and that it is a projector and so $(\Lambda^+)^2=\Lambda^+$. Finally, the factor 
    \begin{equation}
        f_\sigma=\frac{1}{2^4}\sum_{P_1,P_2\in \mathcal{P}_1}  \tr[P_\theta^{\ot 4} P_1^{\ot 4} (P_\theta^\dagger)^{\ot 4} P_2^{\ot 4} T_{\sigma} ]
    \end{equation}
    is found from direct computation.
\end{proof}

\section{Average operator purity over the $\nu$-compressible ensemble} \label{ap:nu_lowerbound}
 We first consider the $\nu$-compressible ensemble, as it is somewhat simpler to compute than the $T$-doped averages. Recall that the $\nu$-compressible ensemble is defined as the class of unitaries  $\{ U_\nu  = C_0 (V \otimes \id_{N-\nu}) C_1: C_{0/1} \in \mc{C}_N,  V \in \mathcal{U}_\nu \}$, equipped with the uniform measure over $C_{0},C_1 \in\mc{C}_N$ and $V \in \mc{U}_{\lceil \nu/2 \rceil}$ on the Clifford and unitary groups respectively. For ease of notation, we will write $V \equiv  (V \otimes \id_{N-\lceil \nu/2 \rceil})$ and ${\tilde{\nu}} := \lceil \nu/2 \rceil$ in the following. As usual, we take initial operator to be some non-identity Pauli string, $O \in \mathcal{P}_N \backslash \{ \id \}$.

Consider the average operator purity over the $\nu$-compressible ensemble. Substituting $U=C_0VC_1$, we have that 
\begin{equation}
    \int_{C_0,C_1 \in \mc{C}, V \in \mc{U}} E^{(\text{pur})}_A (O_U)= \frac{1}{D^2}\tr\left[\left( \int_{C_0,C_1 \in \mc{C}, V \in \mc{U}} O_U^{\ot 4} \right)T'\right] \label{eq:purityLOEav}
\end{equation}
where 
\begin{equation}
    O_U =C_1^\dg V^\dg C_0^\dg O C_0VC_1,
\end{equation}
and $T'$ is defined in Eq.~\eqref{eq:Tprime}. As each average is independent, we will first perform the Clifford average over $C_0$. From Eq.~\eqref{eq:cliff_av}, we find that 
\begin{align}
    \int_{C_0 \in \mc{C}_N} (C_0^\dg)^{\ot 4} O^{\ot 4} C_0^{\ot 4} &=\sum_{\zeta, \eta \in \mathcal{S}_4} \mathrm{Wg}^+_{\zeta \eta} \tr[O^{\ot 4} \Lambda^+ T_\zeta]\Lambda^+ T_\eta + \mathrm{Wg}^-_{\zeta \eta} \tr[O^{\ot 4} \Lambda^{- } T_\zeta] \Lambda^{- } T_\eta \\
    &=\sum_{\zeta, \eta \in \mathcal{S}_4} \mathrm{Wg}^+_{\zeta \eta} D^{\# (\zeta) -2\delta_{\zeta o}} \Lambda^+ T_\eta - \mathrm{Wg}^-_{\zeta \eta} \delta_{\zeta o} D^{\# (\zeta) -2} \Lambda^{- } T_\eta \\
    &= \sum_{\zeta, \eta \in \mathcal{S}_4} (\mathrm{Wg}^+_{\zeta \eta}  + \delta_{\zeta o} \mathrm{Wg}^-_{\zeta \eta}) D^{\# (\zeta) -2\delta_{\zeta o}} \Lambda^+ T_\eta - \delta_{\zeta o} \mathrm{Wg}^-_{\zeta \eta}  D^{\# (\zeta) -2}  T_\eta, \label{eq:gsd}
\end{align}
where we have used Lemma~\ref{lem:1}. 

For the next averaging over $V \in \mc{U}_{\tilde{\nu}}$, we need only consider its effect on the terms $\Lambda^+ T_\eta$ and $T_\eta$ from the above Eq.~\eqref{eq:gsd}, and sub back in the constants and summations at the end. However, the unitary $V$ only acts on ${\tilde{\nu}}$ qubits, with identity elsewhere. Both $\Lambda^+ T_\eta$ and $T_\eta$ are product operators within a single replica space, and so 
\begin{equation}
    \int_{V \in \mathbb{H}} (\id \ot V^\dg )^{\ot 4} \Lambda^+ T_\eta (\id \ot V^{\ot 4}) = \Lambda^+ T_\eta \otimes \int_{V \in \mathbb{H}} (V^\dg)^{\ot 4} \Lambda^+ T_\eta V^{\ot 4},
\end{equation}
where the dimensionality of each $\Lambda^+$, $T_\eta$ is clear from context: the first term of the tensor product in the final expression acts on $N-{\tilde{\nu}}$ qubits across $4-$replica space, while the second term they act on the remaining ${\tilde{\nu}}$ qubits. Note that the exact identity of this bipartition is irrelevant, as the global Clifford averaging is invariant under left/right multiplication by permutations. 

We can then use the usual Haar integration formula, Eq.~\eqref{eq:haarWein}, defining $D_{\tilde{\nu}} := 2^{{\tilde{\nu}}}$ we find that 
\begin{align}
     &\int_{V \in \mathbb{H}} (V^\dg)^{\ot 4} \Lambda^+ T_\eta V^{\ot 4} =  \sum_{\kappa, \mu \in \mathcal{S}_4}  \mathrm{Wg}_{\kappa \mu} (D_{\tilde{\nu}}) \tr[ \Lambda^+ T_\eta T_\kappa ] T_\mu = \sum_{\kappa, \mu \in \mathcal{S}_4}  \mathrm{Wg}_{\kappa \mu} D_{\tilde{\nu}}^{\#(\eta \kappa)-2\delta_{(\eta \kappa) o }} T_\mu, \text{ and} \\
     &\int_{V \in \mathbb{H}} (V^\dg)^{\ot 4}  T_\eta V^{\ot 4} =  \sum_{\kappa, \mu \in \mathcal{S}_4}  \mathrm{Wg}_{\kappa \mu} (D_{\tilde{\nu}}) \tr[  T_\eta T_\kappa ] T_\mu = \sum_{\kappa, \mu \in \mathcal{S}_4}  \mathrm{Wg}_{\kappa \mu} D^{\#(\eta \kappa)}_{\tilde{\nu}} T_\mu \label{eq:HaarPart}
\end{align}
where we have again used Lemma~\ref{lem:1}. We have here in the first equality explicitly included the dimensional dependence of the Weingarten functions, as this Haar average is over a space of ${\tilde{\nu}}$ qubits only. In the following, all Weingarten functions $\mathrm{Wg}_{\pi} = \mathrm{Wg}_{\pi}(D_{\tilde{\nu}})$, whereas all generalized Clifford Weingarten functions have the dependence $\mathrm{Wg}^\pm_{\pi} = \mathrm{Wg}_{\pi}^\pm(D)$. Substituting this into Eq.~\eqref{eq:gsd}, we have that 
\begin{align}
     \int_{C_0 \in \mc{C}_N, V \in \mc{U}_{\tilde{\nu}}} &(V^\dagger)^{\ot 4} (C_0^\dg)^{\ot 4} O^{\ot 4} C_0^{\ot 4} V^{\ot 4} \\
     =& \sum_{\zeta, \eta \in \mathcal{S}_4} (\mathrm{Wg}^+_{\zeta \eta}  + \delta_{\zeta o} \mathrm{Wg}^-_{\zeta \eta}) D^{\# (\zeta) -2\delta_{\zeta o}} \Lambda^+ T_\eta \otimes (\sum_{\kappa, \mu \in \mathcal{S}_4}  \mathrm{Wg}_{\kappa \mu} D_{\tilde{\nu}}^{\#(\eta \kappa)-2\delta_{(\eta \kappa) o }} T_\mu) \\
     &- \delta_{\zeta o} \mathrm{Wg}^-_{\zeta \eta}  D^{\# (\zeta) -2}  T_\eta \otimes (\sum_{\kappa, \mu \in \mathcal{S}_4}  \mathrm{Wg}_{\kappa \mu} D^{\#(\eta \kappa)}_{\tilde{\nu}} T_\mu)\\
     =& \sum_{\zeta, \eta,\kappa, \mu \in \mathcal{S}_4} (\mathrm{Wg}^+_{\zeta \eta}  + \delta_{\zeta o} \mathrm{Wg}^-_{\zeta \eta}) \mathrm{Wg}_{\kappa \mu} D^{\# (\zeta) -2\delta_{\zeta o}}  D^{\#(\eta \kappa)-2\delta_{(\eta \kappa) o }}_{\tilde{\nu}} (\Lambda^+ T_\eta \otimes T_\mu) \\
     &- \delta_{\zeta o} \mathrm{Wg}^-_{\zeta \eta}  \mathrm{Wg}_{\kappa \mu} D^{\# (\zeta) -2} D^{\#(\eta \kappa)}_{\tilde{\nu}}  (T_\eta \otimes T_\mu).\label{eq:hshsh}
\end{align}
Here, we stress again that the tensor product $T_\eta \otimes T_\mu$ is (in-order) on the first $N-{\tilde{\nu}}$ qubits and the next ${\tilde{\nu}}$ qubits respectively (in $4-$replica space). 

Now we will perform the final Clifford averaging, over the operators $(\Lambda^+ T_\eta \otimes T_\mu) $ and $(T_\eta \otimes T_\mu)$ in Eq.~\eqref{eq:hshsh}. We first find that 
\begin{align}
    \int_{C_1 \in \mc{C}_N} (C_1^\dg)^{\ot 4} (T_\eta \otimes T_\mu) C_1^{\ot 4} =&\sum_{\pi, \sigma \in \mathcal{S}_4} \mathrm{Wg}^+_{\pi \sigma} \tr[(T_\eta \otimes T_\mu ) \Lambda^+ T_\pi]\Lambda^+ T_\sigma + \mathrm{Wg}^-_{\pi \sigma} \tr[(T_\eta \otimes T_\mu) \Lambda^{- } T_\pi] \Lambda^{- } T_\sigma \nn \\
    =& \sum_{\pi, \sigma \in \mathcal{S}_4} \mathrm{Wg}^+_{\pi \sigma} D_{N-{\tilde{\nu}}}^{\# (\pi \eta) -2 \delta_{(\pi \eta)o}}D_{{\tilde{\nu}}}^{\# (\pi \mu) -2 \delta_{(\pi \mu)o}} \Lambda^+ T_\sigma \nn  \\
    &+ \mathrm{Wg}^-_{\pi \sigma} (\tr[(T_\eta \otimes T_\mu) T_\pi] -\tr[(T_\eta \otimes T_\mu) \Lambda^+ T_\pi])(\id - \Lambda^+) T_\sigma\nn  \\
    =& \sum_{\pi, \sigma \in \mathcal{S}_4} \mathrm{Wg}^+_{\pi \sigma} D_{N-{\tilde{\nu}}}^{\# (\pi \eta) -2 \delta_{(\pi \eta)o}}D_{{\tilde{\nu}}}^{\# (\pi \mu) -2 \delta_{(\pi \mu)o}} \Lambda^+ T_\sigma\nn  \\
    &+ \mathrm{Wg}^-_{\pi \sigma} (D_{N-{\tilde{\nu}}}^{\# (\pi \eta)} D_{{\tilde{\nu}}}^{\# (\pi \mu)} - D_{N-{\tilde{\nu}}}^{\# (\pi \eta) -2 \delta_{(\pi \eta)o}} D_{{\tilde{\nu}}}^{\# (\pi \mu) -2 \delta_{(\pi \mu)o}})(\id - \Lambda^+) T_\sigma\nn  \\
    =& \sum_{\pi, \sigma \in \mathcal{S}_4} \left( (\mathrm{Wg}^+_{\pi \sigma} +\mathrm{Wg}^-_{\pi \sigma} )D_{N-{\tilde{\nu}}}^{\# (\pi \eta) -2 \delta_{(\pi \eta)o}}D_{{\tilde{\nu}}}^{\# (\pi \mu) -2 \delta_{(\pi \mu)o}} - \mathrm{Wg}^-_{\pi \sigma} D_{N-{\tilde{\nu}}}^{\# (\pi \eta)} D_{{\tilde{\nu}}}^{\# (\pi \mu)}\right) \Lambda^+ T_\sigma \nn \\
    &+ \mathrm{Wg}^-_{\pi \sigma} (D_{N-{\tilde{\nu}}}^{\# (\pi \eta)} D_{{\tilde{\nu}}}^{\# (\pi \mu)} - D_{N-{\tilde{\nu}}}^{\# (\pi \eta) -2 \delta_{(\pi \eta)o}} D_{{\tilde{\nu}}}^{\# (\pi \mu) -2 \delta_{(\pi \mu)o}}) T_\sigma \label{eq:b43}
\end{align}
In the above, we have used that both $\Lambda^+$ and $T_\pi$ are product operators between qubits: that they factorize into acting independently on the $4-$replica space of the first $N-{\tilde{\nu}}$ qubits, and on the other ${\tilde{\nu}}$ qubits. 

Similarly, 
\begin{align}
    \int_{C_1 \in \mc{C}_N} (C_1^\dg)^{\ot 4} (\Lambda^+ T_\eta \otimes T_\mu) C_1^{\ot 4} =&\sum_{\pi, \sigma \in \mathcal{S}_4} \mathrm{Wg}^+_{\pi \sigma} \tr[(\Lambda^+ T_\eta \otimes T_\mu ) \Lambda^+ T_\pi]\Lambda^+ T_\sigma + \mathrm{Wg}^-_{\pi \sigma} \tr[(\Lambda^+ T_\eta \otimes T_\mu) \Lambda^{- } T_\pi] \Lambda^{- } T_\sigma \nn \\
    =& \sum_{\pi, \sigma \in \mathcal{S}_4} \mathrm{Wg}^+_{\pi \sigma} D_{N-{\tilde{\nu}}}^{\# (\pi \eta) -2 \delta_{(\pi \eta)o}}D_{{\tilde{\nu}}}^{\# (\pi \mu) -2 \delta_{(\pi \mu)o}} \Lambda^+ T_\sigma\nn  \\
    &+ \mathrm{Wg}^-_{\pi \sigma} (D_{N-{\tilde{\nu}}}^{\# (\pi \eta)-2\delta_{(\pi \eta) o}} D_{{\tilde{\nu}}}^{\# (\pi \mu)} - D_{N-{\tilde{\nu}}}^{\# (\pi \eta) -2 \delta_{(\pi \eta)o}} D_{{\tilde{\nu}}}^{\# (\pi \mu) -2 \delta_{(\pi \mu)o}})(\id - \Lambda^+) T_\sigma\nn  \\
    =& \sum_{\pi, \sigma \in \mathcal{S}_4} \left( (\mathrm{Wg}^+_{\pi \sigma} +\mathrm{Wg}^-_{\pi \sigma} )D_{N-{\tilde{\nu}}}^{\# (\pi \eta) -2 \delta_{(\pi \eta)o}}D_{{\tilde{\nu}}}^{\# (\pi \mu) -2 \delta_{(\pi \mu)o}} - \mathrm{Wg}^-_{\pi \sigma} D_{N-{\tilde{\nu}}}^{\# (\pi \eta)-2\delta_{(\pi \eta) o}} D_{{\tilde{\nu}}}^{\# (\pi \mu)}\right) \Lambda^+ T_\sigma \nn \\
    &+ \mathrm{Wg}^-_{\pi \sigma} (D_{N-{\tilde{\nu}}}^{\# (\pi \eta)-2\delta_{(\pi \eta) o}} D_{{\tilde{\nu}}}^{\# (\pi \mu)} - D_{N-{\tilde{\nu}}}^{\# (\pi \eta) -2 \delta_{(\pi \eta)o}} D_{{\tilde{\nu}}}^{\# (\pi \mu) -2 \delta_{(\pi \mu)o}}) T_\sigma \label{eq:b44}
\end{align}
Here, we have used that $\Lambda^+$ is a projector, and so $(\Lambda^+)^2= \Lambda^+$. This means that the above is almost identical to Eq.~\eqref{eq:b43}, expect for the factor proportional to $\tr[(\Lambda^+ T_\eta \otimes T_\mu) T_\pi] $.

Now, it remains to substitute the full (triple-)average over Clifford, then unitary, then Clifford groups, into the expression for operator purity, Eq.~\eqref{eq:purityLOEav}. Recalling the definition of the purity-permutation $T'$ in Eq.~\eqref{eq:Tprime}, from Eq.~\eqref{eq:b44} we simply need to compute $\tr[T_\sigma T']$ and $\tr[\Lambda^+  T_\sigma T']$, as everything else in our expression is a constant. Again utilizing that permutations and the projector $\Lambda^+$ are product operators between qubits, we find that, 
\begin{align}
     \tr[\Lambda^+ T_\sigma T' ] &= \frac{1}{D^2}\sum_{\substack{P_1 \in \mc{P}_{A} \\ P_2 \in \mc{P}_{\bar{A}}}} \tr[P_1 T_\sigma T_{(12)(34)}] \tr[P_2 T_\sigma T_{(14)(23)}] \\
     &= \begin{cases}
        D_{A}^{\# (\sigma {(12)(34)})} D_{\bar{A}}^{\#(\sigma {(14)(23)})} & \text{for } \sigma {(12)(34)}, \, \sigma {(14)(23)} \in \mc{S}^e \;,\\
        D_{A}^{\# (\sigma {(12)(34)})} D_{\bar{A}}^{\#(\sigma {(14)(23)})-2}  & \text{for } \sigma {(12)(34)} \in \mc{S}^e, \sigma {(14)(23)} \in \mc{S}^o\;,\\
        D_{A}^{\# (\sigma {(12)(34)})-2} D_{\bar{A}}^{\#(\sigma {(14)(23)})}  & \text{for } \sigma {(12)(34)} \in \mc{S}^o, \sigma {(14)(23)} \in \mc{S}^e\;,\\
        D_{A}^{\# (\sigma {(12)(34)})-2} D_{\bar{A}}^{\#(\sigma {(14)(23)})-2}   & \text{for }  \sigma {(12)(34)}, \,\sigma {(14)(23)} \in \mc{S}^o \;.
        \end{cases} \label{eq:lamTT} \\
        &=D_{A}^{\# (\sigma {(12)(34)}) -2\delta_{(\sigma {(12)(34)}) o}} D_{\bar{A}}^{\#(\sigma {(14)(23)}) -2\delta_{(\sigma {(14)(23)}) o}}.
\end{align}
Similarly, 
\begin{equation}
    \tr[T_\sigma T' ] = D_A^{\#( \sigma (12)(34) )}D_{\bar{A}}^{\#(\sigma (14)(23) )}. \label{eq:TT}
\end{equation}
To obtain our final expression, we use Eqs.~\eqref{eq:lamTT}-\eqref{eq:TT} to give the action of the final trace with respect to $T'$ in Eqs.~\eqref{eq:b43}-\eqref{eq:b44}, and then use these to give the final $4-$fold Clifford averaging of Eq.~\eqref{eq:hshsh}. Remembering the $1/D^2$ overall normalization in the operator purity Eq.~\eqref{eq:purityLOE}, we arrive at 
\begin{align}
    \int_{C_0,C_1 \in \mc{C}, V \in \mc{U}} E^{(\text{pur})}_A =& \frac{1}{D^2}\sum_{\zeta, \eta,\kappa, \mu,\pi,\sigma \in \mathcal{S}_4} (\mathrm{Wg}^+_{\zeta \eta}  + \delta_{\zeta o} \mathrm{Wg}^-_{\zeta \eta}) \mathrm{Wg}_{\kappa \mu} D^{\# (\zeta) -2\delta_{\zeta o}}  D^{\#(\eta \kappa)-2\delta_{(\eta \kappa) o }}_{\tilde{\nu}} \nn \\
    &\times \Bigg( \left( (\mathrm{Wg}^+_{\pi \sigma} +\mathrm{Wg}^-_{\pi \sigma} )D_{N-{\tilde{\nu}}}^{\# (\pi \eta) -2 \delta_{(\pi \eta)o}}D_{{\tilde{\nu}}}^{\# (\pi \mu) -2 \delta_{(\pi \mu)o}} - \mathrm{Wg}^-_{\pi \sigma} D_{N-{\tilde{\nu}}}^{\# (\pi \eta)-2\delta_{(\pi \eta) o}} D_{{\tilde{\nu}}}^{\# (\pi \mu)}\right) \nn\\
    &\times D_{A}^{\# (\sigma {(12)(34)}) -2\delta_{(\sigma {(12)(34)}) o}} D_{\bar{A}}^{\#(\sigma {(14)(23)}) -2\delta_{(\sigma {(14)(23)}) o}} \nn \\
    &+ \mathrm{Wg}^-_{\pi \sigma} (D_{N-{\tilde{\nu}}}^{\# (\pi \eta)-2\delta_{(\pi \eta) o}} D_{{\tilde{\nu}}}^{\# (\pi \mu)} - D_{N-{\tilde{\nu}}}^{\# (\pi \eta) -2 \delta_{(\pi \eta)o}} D_{{\tilde{\nu}}}^{\# (\pi \mu) -2 \delta_{(\pi \mu)o}}) D_A^{\#( \sigma (12)(34) )}D_{\bar{A}}^{\#(\sigma (14)(23) )}\Bigg) \nn\\
     &- \delta_{\zeta o} \mathrm{Wg}^-_{\zeta \eta}  \mathrm{Wg}_{\kappa \mu} D^{\# (\zeta) -2} D^{\#(\eta \kappa)}_{\tilde{\nu}}  \Bigg(\Big( (\mathrm{Wg}^+_{\pi \sigma} +\mathrm{Wg}^-_{\pi \sigma} )D_{N-{\tilde{\nu}}}^{\# (\pi \eta) -2 \delta_{(\pi \eta)o}}D_{{\tilde{\nu}}}^{\# (\pi \mu) -2 \delta_{(\pi \mu)o}} \nn\\
     &- \mathrm{Wg}^-_{\pi \sigma} D_{N-{\tilde{\nu}}}^{\# (\pi \eta)} D_{{\tilde{\nu}}}^{\# (\pi \mu)}\Big) D_{A}^{\# (\sigma {(12)(34)}) -2\delta_{(\sigma {(12)(34)}) o}} D_{\bar{A}}^{\#(\sigma {(14)(23)}) -2\delta_{(\sigma {(14)(23)}) o}} \nn \\
    &+ \mathrm{Wg}^-_{\pi \sigma} (D_{N-{\tilde{\nu}}}^{\# (\pi \eta)} D_{{\tilde{\nu}}}^{\# (\pi \mu)} - D_{N-{\tilde{\nu}}}^{\# (\pi \eta) -2 \delta_{(\pi \eta)o}} D_{{\tilde{\nu}}}^{\# (\pi \mu) -2 \delta_{(\pi \mu)o}}) D_A^{\#( \sigma (12)(34) )}D_{\bar{A}}^{\#(\sigma (14)(23) )}\Bigg) \\
    =& \frac{16^{-{\tilde{\nu}}}}{\left(-9 + 4^{\tilde{\nu}}\right) \left(-4 + D^2\right) \left(-1 + D^2\right)^2 D_A^2} \Bigg( 24 D^4 \left(d - d_A\right) \left(d + d_A\right) \left(-1 + d_A^2\right) \nn\\
& - 4^{1 + {\tilde{\nu}}} D^2 \left(-4 + 9 D^2\right) \left(D- D_A\right) \left(D + D_A\right) \left(-1 + D_A^2\right) \nn\\
& + 16^{\tilde{\nu}} \big( 18 D^4 - 13 D^6 + \left(36 - 99 D^2 + 49 D^4 + 4 D^6\right) D_A^2 + 
D^2 \left(18 - 13 D^2\right) D_A^4
\big) \nn\\
& + 64^{\tilde{\nu}} \big(
D^6 - 4 D_A^2 + D^2 D_A^2 \left(11 - 2 D_A^2\right) + 
D^4 \left(-2 - 5 D_A^2 + D_A^4\right)
\big)
\Bigg).\nn
\end{align}
This was computed symbolically in Mathematica. The above is a generalization of the top upper bound presented in Thm.~\ref{thm:main_lower}; it is valid for any $D$, $\nu$, and $D_A$.
We will now explore various limits of this quantity to extract some physical insight.

First consider ${\tilde{\nu}}=N$. Our answer should agree with the Haar averaged LOE, as the Cliffords $C_0,C_1$ will have no effect in this case: Clifford twirling a Haar twirl leaves the Haar twirl invariant, due to the unitary invariance of the Haar measure.
\begin{cor} \label{cor:haar1}
For a Haar random unitary $V$, the average Heisenberg operator purity is 
    \begin{align}
        \int_{V \in \mc{U}_N} E^{(\text{pur})}_A(V^\dagger O V) = \lim_{{\tilde{\nu}} \to N}\int_{ C_0,C_1 \in \mc{C}_N, V \in \mc{U}_{\tilde{\nu}}} E^{(\text{pur})}_A &= \frac{  
  (D^4_A+D^2) (D^2-10) -D_A^2 (D^2-19  ) }{ D_A^{2}(D^2-9)(D^2-1)} \\
  &= \frac{1}{D_A^2} + \frac{1}{D_{\bar{A}}^2} + \mc{O}(\frac{1}{D}).
    \end{align}
    where the final equality is written to leading order in $1/D$.
\end{cor}
Note that the final expression agrees with the leading order term for the Haar averaged operator purity from Ref.~\cite{Kudler-Flam2021} (see Eq.~(29)). Moreover, choosing an equal-sized bipartition of $D_A = \sqrt{D}$, we arrive at Eq.~\eqref{eq:haar} in the main text, 
\begin{equation}
    \int_{U \sim \mc{U}_N} E^{(\text{pur})}_{A}(O_U) = \frac{-19 + D + 2 D^2}{(1 + D) (-9 + D^2)}. \label{eq:haar1}
\end{equation}

Now consider the limit $D_A = \sqrt{D}$ for arbitrary $\nu$, to get the simpler expression,
\begin{align}
    \int_{C_0,C_1 \in \mc{C}_N, V \in \mc{U}_{\tilde{\nu}}} E^{(\text{pur})}_{N_A=N/2}(O_U) = & \frac{16^{-{\tilde{\nu}}}}{\left(-9 + 4^{\tilde{\nu}}\right) \left(1 + D\right)^2 \left(-4 + D^2\right)} \Bigg(  24 D^4 - 4^{1 + {\tilde{\nu}}} D^2 \left(-4 + 9 D^2\right) \\
& + 64^{\tilde{\nu}} \big(-4 + D \left(-8 + D \left(-1 + 2D\right)\right)\big) + 16^{\tilde{\nu}} \big(
36 + D \left(72 + D \left(9 + 2D \left(-9 + 2D\right)\right)\right) \nn
\big)
\Bigg).
\end{align}
Taking the leading and subleading order expression in large $D$, we arrive at
\begin{align}
    \int_{C_0,C_1 \in \mc{C}_N, V \in \mc{U}_{\tilde{\nu}}} E^{(\text{pur})}_{N_A=N/2} (O_U) &\approx
    \frac{4 D_{\tilde{\nu}}^4-36 D_{\tilde{\nu}}^2 +24}{D_{\tilde{\nu}}^4(D_{\tilde{\nu}}^2-9)}+\frac{2D_{\tilde{\nu}}^2 -18}{D_{\tilde{\nu}}^2-9} \left(  \frac{1}{D}\right)+ \mc{O}\left( \frac{1}{D^2}\right)\\
    &= \frac{4 }{D_{\tilde{\nu}}^2}+ \frac{24}{D_{\tilde{\nu}}^4(D_{\tilde{\nu}}^2-9)} + \frac{2}{D}+ \mc{O}\left( \frac{1}{D^2}\right)
\end{align}
Here, we have neglected factors proportional to $1/D^2$. To obtain a lower bound as that which appears in Thm.~\ref{thm:main_lower}, we apply Jensen's inequality, 
\begin{equation}
     -\log\left[\int_{C_0,C_1 \in \mc{C}_N, V \in \mc{U}_{\tilde{\nu}}} E^{(\text{pur})}_{N_A=N/2} (O_U)\right] \leq  \int_{C \in \mc{C}, V \in \mathbb{H}} E^{(2)}_{N_A=N/2}(O_U)
\end{equation}
and so after some simplification, (also recalling that $D_\nu = 2^{\nu/2}$),
\begin{align}
&\nu -2 - \log\left(1 + \frac{6}{D_{\tilde{\nu}}^4(D_{\tilde{\nu}}^2-9)} + \mc{O}(D_{\tilde{\nu}}^2/D)\right)\leq  \int_{C_0,C_1 \in \mc{C}_N, V \in \mc{U}_{\tilde{\nu}}}  E^{(2)}_{N/2} \\
\implies & \nu -2 - \mc{O}\left( 2^{-2\nu} + 2^{\nu-N }\right) \leq  \int_{C_0,C_1 \in \mc{C}_N, V \in \mc{U}_{\nu/2}}  E^{(2)}_{N_A=N/2}
\end{align}
where the final line is an approximation for small $ \frac{D_\nu^2}{D} + \frac{1}{D_\nu^2}$. This is our final bound in terms of unitary nullity.

\section{Average operator purity over the $T$-doped Clifford ensemble} \label{ap:tau_lower-bounds}
We first restate the result from Ref.~\cite{Leone2021quantumchaosis} (based on results from \cite{zhu2016cliffordgroupfailsgracefully}):
\begin{thm} (\cite{Leone2021quantumchaosis}) \label{lem:lorenzo}
    For the set of $T$-doped Clifford circuits $\mc{C}_\tau$ with $\tau$ $T$-gates, the uniform $4-$fold average over some four-replica quantity $X \in \mc{H}^{\ot 4}$ is
    \begin{equation}
    \Phi^{(4)}_{\mc{C}_\tau}(X) = \sum_{\pi, \sigma \in\mc{S}_4} 
    \left[ \left( (\Xi^\tau)_{\pi \sigma} \Lambda^+ + \Gamma^{(\tau)}_{\pi \sigma} \right) c_{\pi}(X) + \delta_{\pi \sigma} b_{\pi}(X) \right] T_{\sigma} \label{eq:lorenzo}
\end{equation}
where $\Lambda^\pm$ are defined as in Eq.~\eqref{eq:lambdas}, 
\begin{equation}
    \Xi_{\sigma \pi} \equiv \sum_{\mu \in S_4} \left[ \mathrm{Wg}^+_{\pi \mu} \mathrm{tr} \left( T_{\sigma} K^{\otimes 4} \Lambda^+ K^{\dagger \otimes 4} \Lambda^+ T_{\mu} \right) - \mathrm{Wg}^-_{\pi \mu} \mathrm{tr} \left( T_{\sigma} K^{\otimes 4} \Lambda^+ K^{\dagger \otimes 4} \Lambda^{-} T_{\mu} \right) \right]
\end{equation}
with
\begin{equation}
    \Gamma^{(\tau)}_{\pi \sigma} = \sum_{\mu \in S_4} F_{\pi \mu} \sum_{i=0}^{\tau-1} (\Xi^i)_{\mu \sigma}
\end{equation}
\begin{equation}
    F_{\pi \mu} = \sum_{\sigma \in S_4} \mathrm{Wg}^-_{\pi \sigma} \mathrm{tr} \left(  K^{\otimes 4} \Lambda^+ K^{\dagger \otimes 4} \Lambda^{-} T_{\sigma \mu} \right)
\end{equation}
and the information about $X$ is contained in the coefficients (from the first Clifford averaging)
\begin{equation}
    c_{\pi}(X) = \sum_{\sigma \in S_4} \left[ \mathrm{Wg}^+_{\pi \sigma} \mathrm{tr}(X \Lambda^+ T_{\sigma}) - \mathrm{Wg}^-_{\pi \sigma} \tr[X \Lambda^{-} T_{\sigma}] \right]
\end{equation}
\begin{equation}
    b_{\pi}(X) = \sum_{\sigma \in S_4} \mathrm{Wg}^-_{\pi \sigma} \mathrm{tr}(X \Lambda^{-} T_{\sigma}).
\end{equation}
\end{thm}

For our purposes, the quantity of interest is the operator purity, 
\begin{equation}
    \int_{U \in \mc{C}_\tau} E^{(\text{pur})}_A = \frac{1}{D^2} \int_{U \in \mc{C}_\tau} \tr[ (U^\dg)^{\otimes 4} O^{\ot 4} U^{\ot 4} T'] = \frac{1}{D^2}  \tr[ \Phi^{(4)}_{\mc{C}_\tau}(O^{\ot 4}) T'].
\end{equation}
As earlier, the first step involves a global Clifford average, so the identity of the initial operator $O$ (beyond it being a traceless Pauli string) is irrelevant. As a first step, we simplify the coefficients using Lemma~\ref{lem:1}
\begin{align}
    c_\pi (O^{\ot 4} ) &= \sum_{\sigma \in S_4} \left[ \mathrm{Wg}^+_{\pi \sigma} \mathrm{tr}(O^{\ot 4} \Lambda^+ T_{\sigma}) - \mathrm{Wg}^-_{\pi \sigma} \mathrm{tr}(O^{\ot 4} \Lambda^{-} T_{\sigma}) \right] \\
    &= \sum_{\sigma \in S_4} \left[ (\mathrm{Wg}^+_{\pi \sigma} + \mathrm{Wg}^-_{\pi \sigma}  ) \frac{1}{D^2} \sum_{P \in \mc{P}_N} \mathrm{tr}(O^{\ot 4} P^{\ot 4} T_{\sigma}) - \mathrm{Wg}^-_{\pi \sigma} \mathrm{tr}(O^{\ot 4} T_{\sigma}) \right]\\
    &= \sum_{\sigma \in S_4} \left[ (\mathrm{Wg}^+_{\pi \sigma} + \mathrm{Wg}^-_{\pi \sigma}  ) D^{\#(\sigma)-2\delta_{\sigma o}} - \mathrm{Wg}^-_{\pi \sigma} \delta_{\sigma e} D^{\#(\sigma)} \right],
\end{align}
where we have used the $\delta_{\sigma o}$ notation defined in Eq.~\eqref{eq:deltas}. Similarly, 
\begin{equation}
    b_{\pi}(O^{\ot 4} )= \sum_{\sigma \in S_4} \left[ \mathrm{Wg}^-_{\pi \sigma} \delta_{\sigma e} D^{\#(\sigma)}-\mathrm{Wg}^-_{\pi \sigma}   D^{\#(\sigma)-2\delta_{\sigma o}} \right] = -\sum_{\sigma \in S_o} \mathrm{Wg}^-_{\pi \sigma}   D^{\#(\sigma)-2}.
\end{equation}
Now, for the matrix $F$, we assume that $K$ is a single qubit $R^Z_\theta $ gate wlog acting on the first qubit, which we call $P_\theta$. For the rest of the proof, we will use the Lemmas~\ref{lem:1}-\ref{lem:2}.

Now, we will consider each of the factors in Lemma~\ref{lem:lorenzo}. Directly applying the Lemma~\ref{lem:2} to the matrix $F$, we find
\begin{align}
    F_{\pi \mu} =& \sum_{\sigma \in S_4} \mathrm{Wg}^-_{\pi \sigma} \mathrm{tr} \left(  K^{\otimes 4} \Lambda^+ K^{\dagger \otimes 4} \Lambda^{-} T_{\sigma \mu} \right) \\
    =&\sum_{\sigma \in S_4} \mathrm{Wg}^-_{\pi \sigma} \left( \mathrm{tr}  [ K^{\otimes 4} \Lambda^+ K^{\dagger \otimes 4} T_{\sigma \mu} ] - \mathrm{tr} [ K^{\otimes 4} \Lambda^+ K^{\dagger \otimes 4} \Lambda^+ T_{\sigma \mu} ]\right) \\
    =&\sum_{\sigma \in S_4} \mathrm{Wg}^-_{\pi \sigma} (  D^{\#(\sigma \mu) -2 \delta_{(\sigma \mu) o}} -D^{\#(\sigma \mu) -2 \delta_{(\sigma \mu) o}}_{N-1} f_{(\sigma \mu)} ).
\end{align}
The factor $f_{(\sigma \mu)}$ can be found through direct calculation, as it is only a single-qubit trace equation. Then, also directly, 
\begin{align}
    \Xi_{\sigma \pi} &= \sum_{\mu \in S_4} \left[ \mathrm{Wg}^+_{\pi \mu} \mathrm{tr} \left(  K^{\otimes 4} \Lambda^+ K^{\dagger \otimes 4} \Lambda^+ T_{\mu \sigma} \right) - \mathrm{Wg}^-_{\pi \mu} \mathrm{tr} \left(  K^{\otimes 4} \Lambda^+ K^{\dagger \otimes 4} \Lambda^{-} T_{\mu \sigma} \right) \right]\\
    &= \sum_{\mu \in S_4} \left[ (\mathrm{Wg}^+_{\pi \mu} + \mathrm{Wg}^-_{\pi \mu})\mathrm{tr} \left(  K^{\otimes 4} \Lambda^+ K^{\dagger \otimes 4} \Lambda^+ T_{\mu \sigma} \right) - \mathrm{Wg}^-_{\pi \mu} \mathrm{tr} \left(  K^{\otimes 4} \Lambda^+ K^{\dagger \otimes 4} T_{\mu \sigma} \right) \right] \\
    &=\sum_{\mu \in S_4} \left[ (\mathrm{Wg}^+_{\pi \mu} + \mathrm{Wg}^-_{\pi \mu})D^{\#(\mu \sigma ) -2 \delta_{(\mu \sigma) o}}_{N-1} f_{(\mu \sigma)} - \mathrm{Wg}^-_{\pi \mu} D^{\#(\mu \sigma ) -2 \delta_{(\mu \sigma) o}} \right].
\end{align}
We can exactly diagonalize this $24 \times  24$ matrix symbolically to perform arbitrary matrix powers in the full expression \eqref{eq:lorenzo}. 

Finally, for the `outside' terms, we have the factors 
\begin{equation}
    \tr[\Lambda^+ T_\sigma T'] = D_{A}^{\# (\sigma {(12)(34)}) -2\delta_{(\sigma {(12)(34)}) o}} D_{\bar{A}}^{\#(\sigma {(14)(23)})-2\delta_{(\sigma {(14)(23)}) o}}
\end{equation}
and 
\begin{equation}
    \tr[T_\sigma T'] = D_{A}^{\# (\sigma {(12)(34)})} D_{\bar{A}}^{\#(\sigma {(14)(23)})},
\end{equation}
which we solved for in Eqs.~\eqref{eq:lamTT}-\eqref{eq:TT}.

Putting this all together, we have
 \begin{align}
    &\int_{U \in \mc{C}_\tau} E^{(\text{pur})}_A = \frac{1}{D^2} \tr[ \sum_{\pi, \sigma \in\mc{S}_4} 
    \left[ \left( (\Xi^\tau)_{\pi \sigma} \Lambda^+ + \Gamma^{(\tau)}_{\pi \sigma} \right) c_{\pi}(X) + \delta_{\pi \sigma} b_{\pi}(X) \right] T_{\sigma} T' ]   \nn\\
    &\quad =\frac{1}{D^2} \Big( \sum_{\pi, \sigma \in\mc{S}_4} 
    \Big[ \big( (\Xi^\tau)_{\pi \sigma} \underbrace{D_{A}^{\# (\sigma {(12)(34)}) -2\delta_{(\sigma {(12)(34)}) o}} D_{\bar{A}}^{\#(\sigma {(14)(23)})-2\delta_{(\sigma {(14)(23)}) o}}}_{ \tr[\Lambda^+  T_{\sigma} T' ]}  \nn\\
    &\quad  + \sum_{\kappa \in S_4} \underbrace{\sum_{\zeta \in S_4} \mathrm{Wg}^-_{\pi \zeta} (  D^{\#(\zeta \kappa) -2 \delta_{(\zeta \kappa) o}} -D^{\#(\zeta \kappa) -2 \delta_{(\zeta \kappa) o}}_{N-1} f_{(\zeta \kappa)} ) }_{F_{\pi \kappa}} \sum_{i=0}^{\tau-1} (\Xi^i)_{\kappa \sigma} \underbrace{D_{A}^{\# (\sigma {(12)(34)})} D_{\bar{A}}^{\#(\sigma {(14)(23)})}}_{\tr[T_\sigma T']} \Big] \label{eq:lorenzo1} \\
    &\quad \times \underbrace{\sum_{\mu \in S_4} \left[ (\mathrm{Wg}^+_{\pi \mu} + \mathrm{Wg}^-_{\pi \mu}  ) D^{\#(\mu)-2\delta_{\mu o}} - \mathrm{Wg}^-_{\pi \mu} \delta_{\mu e} D^{\#(\mu)} \right]}_{c_\pi(O^{\ot 4})} - \delta_{\pi \sigma} \underbrace{\left[ \sum_{\kappa \in S_o} \mathrm{Wg}^-_{\pi \kappa}   D^{\#(\kappa)-2} \right] D_{A}^{\# (\sigma {(12)(34)})} D_{\bar{A}}^{\#(\sigma {(14)(23)})}}_{b_\pi(O^{\ot 4}) \tr[T_\sigma T']} \Big) \nn
\end{align}
Evaluating this symbolically, for the case of $N_A=N/2$ (the arbitrary $N_A$ case is too long and complex to report here), we arrive at 
\begin{align}
    \int_{U \in \mc{C}_\tau} E^{(\text{pur})}_A=& \frac{1}{D^2-1} + \frac{D (272 + 48 D^2 + D^4)\left( -(-4 + 3 D(D-1) )^\tau + 4^\tau (D^2-1)^\tau \right) }{4^{1 +\tau} (D^2-1)^\tau (D-1) (D+1) (D+2) (D+3) (D+4)}\\
    &+ \frac{(D-2 ) (D-1 ) (3 D (D+1)-4)^\tau +(D+1) \left((D+2) (3 D (D-1)-4)^\tau +4 (D-1 ) (3 D^2-4)^\tau\right)}{6 \cdot 4^{\tau} (D^2-1)^{\tau  + 1} }.
\end{align}
Expanding the above expression to leading order in large $D$, and taking the $-\log$, we arrive at, 
\begin{align}
    -\log\left(\int_{U \in \mc{C}_\tau} E^{(\text{pur})}_A\right) = -\log\left( (3/4)^\tau - \frac{2}{D}((3/4)^\tau +1) + \mc{O}(3^\tau /D^2 ) \right) = \log({4}/{3}) \tau + \mc{O}\left( \frac{(4/3)^\tau}{D}\right) 
\end{align}
which is valid for $1/D \ll 1$. Applying Jensen's inequality and the hierarchy of R\'enyi entropies (see App.~\ref{ap:loe_purity}), we arrive at the second lower-bound of Thm.~\ref{thm:main_lower}. 

Note that, as expected~\cite{Harrow2009,Leone2021quantumchaosis}, taking the limit of the full expression with $\tau \to \infty$ we arrive at the Haar average value of the operator purity; see Corollary~\ref{cor:haar1} and Eq.~\eqref{eq:haar}.

\end{document}